\documentclass{article}
\usepackage{graphicx}
\usepackage{lineno}
\nolinenumbers
\usepackage{fullpage}
\usepackage{comment}
\usepackage{wrapfig}

\usepackage{multirow}
\usepackage{makecell}

\usepackage[dvipsnames,usenames]{xcolor}
\usepackage{amsmath}
\usepackage{amsfonts}
\usepackage{graphicx}
\usepackage{subcaption}
\usepackage[export]{adjustbox}
\usepackage{wrapfig}
\usepackage{wrapstuff}

\RequirePackage{hyperref}
\hypersetup{colorlinks=true, urlcolor=Blue, citecolor=Green, linkcolor=BrickRed, breaklinks, unicode}

\usepackage{tikz}
\usepackage{xspace,enumerate}
\usepackage{todonotes}
\usepackage{MnSymbol}
\usepackage[noadjust]{cite}
\usepackage{thm-restate}
\usepackage{amsthm}
\usepackage{thmtools, thm-restate}
\usepackage{thm-restate}
\usepackage{authblk}

\input{insbox}

\newcommand{\Left}{\mathsf{Left}}
\newcommand{\leftside}{\mathsf{left}}
\newcommand{\rightside}{\mathsf{right}}
\newcommand{\Vor}{\textsf{Vor}}
\newcommand{\dist}{\mathsf{dist}}
\newcommand{\direction}{\mathsf{direction}}
\newcommand{\countA}{\mathsf{count}}
\newcommand{\select}{\mathsf{select}}
\newcommand{\ancestor}{\mathsf{ancestor}}

\newcommand{\VD}{\textsf{VD}}
\newcommand{\out}[1]{#1^{out}}
\newcommand{\VDin}{\textsf{VD}_{in}}
\newcommand{\VDout}{\textsf{VD}_{out}}
\newcommand{\weight}{{\rm \omega}}
\newcommand{\len}{\mathsf{len}}
\newcommand{\X}{X}

\usepackage{amsthm}
\usepackage{multicol}
\newtheorem{claim}{Claim}
\newtheorem{lemma}{Lemma}
\newtheorem{theorem}{Theorem}
\newtheorem{corollary}{Corollary}
\newtheorem{definition}{Definition}
\newtheorem{remark}{Remark}

\usepackage{amsthm}
\usepackage[capitalize]{cleveref}

\crefname{claim}{Claim}{Claims}

\newcommand{\R}{\mathbb{R}}
\newcommand{\Otild}{\tilde{O}}

\title{Faster Construction of a Planar Distance Oracle \\with $\tilde O(1)$  Query Time\thanks{This research was partially funded by Israel Science Foundation grant 810/21.}
}

\author[1,2]{Itai Boneh}
\author[1,2]{Shay Golan}
\author[1]{Shay Mozes}
\author[1]{Daniel Prigan}
\author[2]{Oren Weimann}

\affil[1]{Reichman University, Israel}
\affil[2]{University of Haifa, Israel}

\date{}

\begin{document}

\maketitle
\vspace{-0.5in}

\begin{abstract}
We show how to preprocess a weighted undirected $n$-vertex planar graph  in $\tilde O(n^{4/3})$ time, such that the distance between any pair of vertices can then be reported in $\tilde O(1)$ time. This improves the previous $\tilde O(n^{3/2})$ preprocessing time [JACM'23].

Our main technical contribution is a near optimal construction of \emph{additively weighted Voronoi diagrams} in undirected planar graphs. Namely, given a planar graph $G$ and a face $f$, we show that one can preprocess $G$ in $\tilde O(n)$ time such that given any weight assignment to the vertices of $f$ one can construct the additively weighted Voronoi diagram of $f$ in near optimal $\tilde O(|f|)$ time. This improves the $\tilde O(\sqrt{n |f|})$ construction time of  [JACM'23].
 \end{abstract}

\section{Introduction}
In this paper we investigate the following question:  {\em How fast can you preprocess a planar graph so that subsequent distance queries can be answered in near-optimal $\tilde O(1)$ time?} Imagine you could do the preprocessing in near-optimal $\tilde O(n)$ time. This would be truly remarkable, since it would allow us to efficiently convert between standard graph representations, which support quick local (e.g., adjacency) queries, into a representation that supports quick non-local distance queries. In a sense, this can be thought of as an analogue for distances of the Fast Fourier Transform (which converts in $\Otild(n)$ time between the time domain and the frequency domain). Such a tool would allow us to develop algorithms for planar graphs that can access any pairwise distance in the graph essentially for free!
Unfortunately, we are still far from obtaining $\tilde O(n)$ preprocessing time. The fastest preprocessing time is  $\tilde O(n^{3/2})$ \cite{ourJACM}. We make a step in this direction by improving the preprocessing time to $\tilde O(n^{4/3})$ in weighted undirected planar graphs.

Formally, a \emph{distance oracle} is a data structure that can report the distance $\dist(u,v)$ between any two vertices $u$ and $v$ in a graph $G$.
Distance oracles for planar graphs have been studied extensively in the past three decades both in the exact~\cite{LongPettie,ourJACM,Fredslund-HMW-N20,ArikatiCCDSZ96,Djidjev96,ChenX00,FakcharoenpholR06,Klein05,Wulff-Nilsen10,Nussbaum11,Cabello12,MozesS2012,Cohen-AddadDW17,GawrychowskiMWW18,CharalampopoulosGMW19} and in the approximate~\cite{Thorup04,KawarabayashiKS11,KawarabayashiST13,Klein02,GuX19,Wulff-Nilsen16,ChanS19,LeW21} settings.
In the approximate setting, Thorup~\cite{Thorup04} presented a near-optimal oracle that returns
 $(1+\varepsilon)$-approximate distances (for any constant $\varepsilon$) in $\tilde O(1)$ time and requires $\tilde O(n)$ space and construction time (see also~\cite{Klein02,KawarabayashiKS11,KawarabayashiST13,GuX19,Wulff-Nilsen16,ChanS19,LeW21} for polylogarithmic improvements). Is it possible that the same near-optimal bounds can be achieved without resorting to approximation? We next review the rich history of {\em exact} oracles in planar graphs.

\medskip
\noindent
{\bf The history of exact planar distance oracles.}
Let $Q$ and $S$ denote the query-time and the space of an oracle, respectively.
The early planar distance oracles~\cite{ArikatiCCDSZ96,Djidjev96,ChenX00} were based solely on {\em planar separators}~\cite{LiptonT80,Miller86,Frederickson87} and achieved a tradeoff of
$Q=\tilde{O}(n/\sqrt{S})$ for $S\in [n^{4/3},n^2]$, and
$Q=O(n^2/S)$ for $S \in [n,n^{4/3})$.
In \cite{FakcharoenpholR06},  Fakcharoenphol and Rao introduced the use of \emph{Monge matrices} to distance oracles, and devised an oracle with $\tilde{O}(n)$ space and $\tilde{O}(\sqrt{n})$ query time.
By combining their ideas with Klein's~\cite{Klein02,CabelloCE13} {\em multiple source shortest path} (MSSP) data structure, Mozes and Sommer~\cite{MozesS2012} obtained the $Q=\tilde{O}(n/\sqrt{S})$ tradeoff for nearly the full range $[n\log\log n, n^2]$.
Other works~\cite{Wulff-Nilsen10,Nussbaum11,MozesS2012,Fredslund-HMW-N20} focused on achieving strictly optimal query-time {\em or} strictly optimal space.
Namely, Wulff-Nilsen's work~\cite{Wulff-Nilsen10} gives optimal $O(1)$ queries with weakly subquadratic $O(n^2\log^4\log n/\log n)$ space, whereas Nussbaum~\cite{Nussbaum11} and Mozes and Sommer's~\cite{MozesS2012} oracles give optimal $O(n)$ space with $O(n^{1/2+\epsilon})$ query-time.
Except for \cite{Wulff-Nilsen10}, all of the above oracles can be  constructed in $\Otild(n)$ time.
However, none of them provides polylogarithmic $\tilde O(1)$ query-time using truly subquadratic $O(n^{2-\varepsilon})$ space.

In FOCS 2017, Cohen-Addad, Dahlgaard, and Wulff-Nilsen~\cite{Cohen-AddadDW17} (inspired by Cabello's~\cite{Cabello19} breakthrough use of \emph{Voronoi diagrams} for computing the diameter of planar graphs) realized that Voronoi diagrams, when applied to regions of an $r$-division, can be used to break the barrier mentioned above.
In particular, they presented the first oracle with $\tilde O(1)$ query-time and truly subquadratic $O(n^{5/3})$ space, or more generally, the tradeoff $Q=\tilde{O}(n^{5/2}/S^{3/2})$ for any $S\in[n^{3/2},n^{5/3}]$.
This came at the cost of increasing the preprocessing time from $\Otild(n)$ to $O(n^2)$, which was subsequently improved to match the space bound~\cite{DBLP:journals/siamcomp/GawrychowskiKMS21}.
In SODA 2018, Gawrychowski, Mozes, Weimann, and Wulff-Nilsen~\cite{GawrychowskiMWW18} improved the space and preprocessing time to $\tilde{O}(n^{3/2})$ with $\tilde O(1)$ query-time (and the tradeoff to $Q=\tilde{O}(n^{3/2}/S)$ for $S\in[n,n^{3/2}]$) by defining a dual representation of Voronoi diagrams and developing an efficient {\em point-location} mechanism on top of it.
In STOC 2019, Charalampopoulos, Gawrychowski, Mozes, and Weimann~\cite{CharalampopoulosGMW19} observed that the same point-location mechanism can be used on the Voronoi diagram for the complement of regions in the $r$-division. This observation alone suffices to improve the oracle size to $O(n^{4/3})$ (while maintaining $\tilde O(1)$ query-time).
By combining this with a sophisticated recursion (where a query at recursion level $i$ reduces to $\log n$ queries at level $i+1$) they further obtained an oracle of size $n^{1+o(1)}$ and query-time $n^{o(1)}$. Finally, in SODA 2021,
Long and Pettie~\cite{LongPettie} showed how much of the point-location work can be done  without recursion, and that only two (rather than $\log n$) recursive calls suffice. This led to the state of the art oracle, requiring $n^{1+o(1)}$ space and $\tilde O(1)$ query-time or $\tilde O(n)$ space and $n^{o(1)}$ query-time.\footnote{A journal version containing all the above Voronoi-based oracles was published in JACM in 2023~\cite{ourJACM}.}

In terms of space and query time, these latter Voronoi-based oracles are almost optimal. However, their construction time is $\tilde O(n^{3/2})$, and improving it is mentioned in \cite{ourJACM} as an important open problem. The bottleneck behind this $\tilde O(n^{3/2})$ bound is the time for constructing Voronoi diagrams as we next explain.

\medskip
\noindent
{\bf Point-location in Voronoi diagrams.}
Let $\X$ be a planar graph, and let $f$ be a face of $\X$. The vertices of $f$ are called the {\em sites} of the Voronoi diagram, and each site $s$ has a weight $\weight(s)\geq 0$ associated with it.
The distance between a site $s$ and a vertex $v\in \X$, denoted by $\dist^\weight(s,v)$ is defined as $\weight(s)$ plus the length of the $s$-to-$v$ shortest path in $\X$.
The {\em additively weighted Voronoi diagram} $\VD(f,\weight)$ is a partition of $\X$'s vertices into pairwise disjoint sets, one set $\Vor(s)$ for each site $s$. The set $\Vor(s)$, called the {\em Voronoi cell} of $s$, contains all vertices in $\X$ that are closer (w.r.t.~$\dist^\weight(\cdot,\cdot)$) to $s$ than to any other site (we assume that  distances are unique to avoid the need to handle ties\footnote{This assumption can be achieved in $O(n)$ time. With high probability using~\cite{Isolation,Isolation2}, or deterministically using~\cite{Isolation3}.}).
A {\em point-location} query $v$ asks for the site $s$ whose Voronoi cell $\Vor(s)$ contains $v$.
There exists a dual representation $\VD^*(f,\weight)$ of $\VD(f,\weight)$. The size of this representation is $O(|f|)$, and together with an MSSP data structure it supports point-location queries in $\Otild(1)$ time.

\begin{figure}[htb]
\begin{minipage}[c]{0.25\textwidth}
    \includegraphics[width=\textwidth]{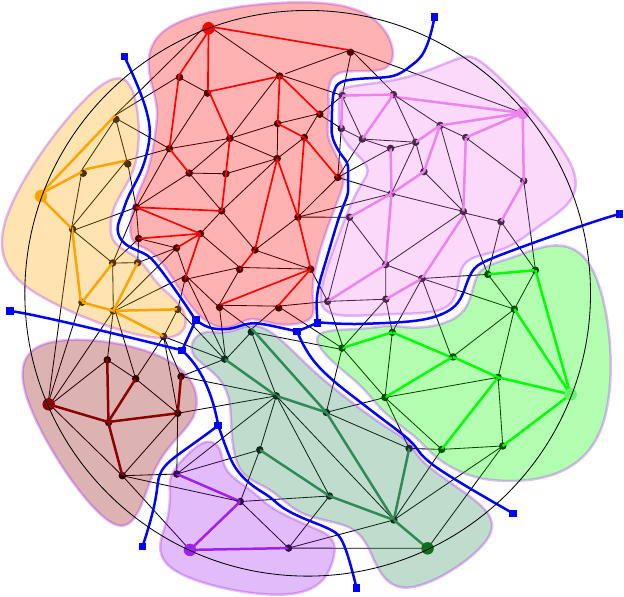}
  \end{minipage}\hfill
  \begin{minipage}[c]{0.7\textwidth}
    \caption{A graph (piece) $\X$. The vertices of $\X$'s infinite face $f$ are the sites of the Voronoi diagram $\VD$. Each site is represented by a unique color, which is also used to shade its Voronoi cell. The dual representation $\VD^*$ is illustrated as the blue tree. The tree has 7 leaves (corresponding to 7 copies of $f^*$) and 5 internal nodes (corresponding to the 5  trichromatic faces of $\VD$). The edges of the tree correspond to contracted subpaths in $\X^*$. \label{fig:VD}}
  \end{minipage}
\end{figure}

\medskip
\noindent
{\bf Our results and techniques.}
In \cite{ourJACM}, it was shown that  given the MSSP of a graph (piece) $\X$, a face $f$ of $\X$, and additive weights $\weight (\cdot)$ to the vertices of $f$, the Voronoi diagram $\VD^*(f,\weight)$ can be constructed in $\tilde O(\sqrt{n |f|})$ time.
In this paper, we show (see \cref{thm:main}) how to construct it in near-optimal $\tilde O(|f|)$ time when the graph is undirected.
As discussed below, this leads to a static distance oracle of space $\tilde O(n^{4/3})$, construction time $\tilde O(n^{4/3})$, and query-time $\tilde O(1)$.
In addition, it implies a {\em dynamic} distance oracle that supports $\tilde O(1)$-time distance queries from a single-source $s$, and $\tilde O(n^{2/3})$-time updates consisting of edge insertions, edge deletions, and  changing the source $s$.
The dynamic distance oracle is obtained by simply plugging  our construction into the oracle of Charalampopoulos and Karczmarz \cite{Panos} (thus improving its update time from $\tilde O(n^{4/5})$ to $\tilde O(n^{2/3})$). We note that the current best {\em all-pairs} dynamic oracle by   Fakcharoenphol and Rao \cite{FakcharoenpholR06} has $\tilde O(n^{2/3})$ update time and $\tilde O(n^{2/3})$ query time. Therefore, in undirected graphs, our new oracle achieves the same bounds but has the benefit that, between source changes, the query takes only $\tilde O(1)$ time.

Our construction is inspired by Charalampopoulos, Gawrychowski, Mozes, and Weimann who showed in ICALP 2021 \cite{CharalampopoulosGM21}  that near-optimal $\tilde O(|f|)$ construction time is possible for the special case when the planar graph is an {\em alignment graph} of two strings.
The alignment graph is an acyclic grid graph of constant degree.
More importantly, shortest paths in this graph are monotone, in the sense that they only go right or down along the grid.
This makes the Voronoi diagram of every piece $\X$ much more structured.
In particular, one can find the trichromatic vertices of $\VD^*$ (i.e., the faces of $\X$ whose vertices belong to three different Voronoi cells) by recursively zooming in on them using cycle separators;
The intersection of a cycle separator with each Voronoi cell consists of at most two contiguous intervals of the cycle separator.
This partition induced by these intervals can be found with a binary search approach, whose complexity (up to a logarithmic factor) is linear in the number of trichromatic vertices enclosed by the cycle.
From the partition it is easy to deduce whether a trichromatic vertex of $\VD^*$ is enclosed by the cycle separator, and infer whether we should recurse on each side of the separator.
In this paper we extend this binary search approach to general undirected planar graphs.

At a first glance, such an extension seems impossible as the intersection of cycle separators and Voronoi cells might be very fragmented, even when there are only three sites (see \cref{fig:fragmented}).
To overcome this, we use {\em shortest path separators}.
Such separators have been widely used in essentially all approximate distance oracles for planar graphs, but also in exact algorithms (e.g.,  \cite{DBLP:conf/soda/ChalermsookFN04,BorradaileSW15,DBLP:conf/esa/LackiS11,DBLP:conf/soda/MozesNNW18}). These cycle separators may contain $\Theta(n)$ edges (rather than $O(\sqrt n)$), but consist of two shortest paths $P,P'$, so that, because the graph is undirected, any other shortest path can intersect each of $P,P'$ at most once.
In our context, this implies a monotonicity property when we restrict our attention to shortest paths that enter the cycle separator only from the left or only from the right.
If we only allow paths to enter $P$ from one side, say the left, then the intersection of each Voronoi cell (with distances now defined under this restriction) with $P$ is a single contiguous interval.

To better explain the difficulties with this approach and how we overcome them,  let us first define some terminology.
We say that vertex $v\in P$ {\em prefers} a site $s$ if $v$ belongs to $\Vor(s)$ (defined without any restrictions).
We say that $v \in P$ left-prefers a site $s$ if $v$ belongs to $\Vor(s)$ under the restriction that paths are only allowed to enter $P$ from the left.
Repeating the above discussion using this terminology, the partition induced by the prefers relation may be very fragmented, but the partition induced by the left-prefers relation is not fragmented.
There are two immediate problems with working with the partition induced by left-preference (or right-preference).
First, we are interested in constructing the Voronoi diagram without any restrictions, and second, we do not know how to efficiently compute shortest paths under restrictions on the direction in which they enter $P$.

To address the first difficulty we show that one can infer which sides of the cycle separator contain a trichromatic vertex of the Voronoi diagram (without restrictions) by inspecting the endpoints of the intervals of the two partitions - the one with respect to left-preference and the one with respect to right-preference. This is shown in \cref{sec:partition_to_trichromatic}. We also show, by adapting the divide-and-conquer construction algorithm of \cite{DBLP:journals/siamcomp/GawrychowskiKMS21}, how to reduce the problem of computing a Voronoi diagram with many sites to the problem of computing the trichromatic vertex of a Voronoi diagram with only three sites. This is explained in \cref{sec:computingVD}.

Overcoming the second difficulty is the heart of our technical contribution.
We know how to efficiently determine,
by enhancing the standard MSSP data structure, whether the true shortest path from a site $s$ to a vertex $v \in P$ enters $P$ from the right or from the left.
However,
we do not know how to efficiently report the shortest path from $s$ to $v$ that enters $P$ from the other side (because that path is not the overall shortest path from $s$ to $v$).
We observe that we can work with {\em relaxed} preferences.
If $v$ prefers $s$ (without restrictions on left/right), then: (1) we might as well say that $v$ both right-prefers and left-prefers $s$ (because either way this will point us to $s$), and (2) if the shortest path from $s$ to $v$ enters $P$ from the right, we do not really care what the left preference of $v$ is, so we might as well say that $v$ left-prefers $s'$ for any site $s'$.
In \cref{sec:relaxed-partition} we describe how to compute a partition with respect to such a relaxed notion of left-preference and right-preference. The exact details are a bit more complicated than described here. In particular we use the term {\em like} rather than {\em prefer} to express that the preference is more loose.

Constructing a partition with relaxed preferences turns out to be more complicated than with strict preferences. One challenge comes from the fact that with relaxed preferences making progress in a binary search procedure is problematic. This is because a vertex $v$ of $P$ no longer has a unique left-preferred site $s$. In fact, when working with relaxed preferences, a vertex $v$ may equally like all sites.
Thus, we need to have some way to recognize a ``winner'' site with respect to a vertex that does not have a unique preference (see \cref{fig:criss-cross}).
Another challenge we encounter is what we call {\em swirly} paths - when the $s$-to-$v$ path goes around $P$ before entering $P$ (see \cref{fig:setup}).
Swirly paths make the realizations of many of the arguments outlined here more complicated.
The structure of the shortest paths that allows us to make progress in the binary search procedure is different and more complicated in the presence of swirly paths.
Moreover, we do not always know how to efficiently identify whether a shortest path is swirly or not.
To simplify the presentation we first explain how to obtain a partition with respect to relaxed preferences when there are no swirly paths, and in \cref{app:two_sites} describe how to handle swirly paths.

We believe that our result is a promising step toward exact oracles that simultaneously have optimal space, query-time and preprocessing time.
One way to achieve this would be to be able to find a way to use recursion to obtain the functionality of our enhanced MSSP data structure without actually computing it over and over in large regions of the graph, in a similar manner to the way this issue was avoided for the standard MSSP in \cite{ourJACM}).
This seems to be the only significant obstacle preventing us from pushing our new ideas all the way through, and obtaining an almost optimal $n^{1+o(1)}$ construction time for undirected planar graphs.

\section{Preliminaries}

\noindent
{\bf Representation of Voronoi diagrams.}
With a standard transformation (without increasing the size of $G$ asymptotically) we can guarantee that each vertex of $G$ has constant degree and that $G$ is triangulated.
Consider a subgraph $\X$ (called a {\em piece}) of $G$ whose boundary vertices $\partial \X$ (vertices incident to edges in $G\setminus \X$) lie on a single face $f$ of $\X$.
Note that $f$ is the only face of $\X$ that is not a triangle.
The vertices of $f$, assigned with weights $\weight(\cdot)$, are the sites of the Voronoi diagram $\VD=\VD(f,\weight)$ of $\X$.
There is a dual representation $\VD^*$ of $\VD$ as a tree with $O(|f|)$ vertices (see \cref{fig:VD}):
Let $\X^*$ be the planar dual of $\X$.
Consider the subgraph of $\X^*$ consisting of the duals of edges $uv$ of $\X$ such that $u$ and $v$ are in different Voronoi cells.
In this subgraph, we repeatedly contract edges incident to degree-2 vertices.
The remaining vertices (edges) are called {\em Voronoi vertices (edges)}.
A Voronoi vertex is dual to a face whose three vertices belong to three different Voronoi cells.
We call such a face (and its corresponding dual vertex) {\em trichromatic}.
Finally, we define $\VD^*$ to be the tree obtained from the subgraph by replacing the node $f^*$ by multiple copies, one for each edge incident to $f^*$.
The complexity (i.e., the number of vertices and edges) of $\VD^*$ is $O(|f|)$.
For example, in a VD of just two sites $s$ and $t$, there are no trichromatic faces.
In this case, $\VD^*$ is just a single edge corresponding to an $st$-cut or equivalently to a cycle in the dual graph.
We call this cycle the $st$-{\em bisector} and denote it by $\beta^*(s,t)$.
Note that a trichromatic vertex is a meeting point of three bisectors.
Another important example is a VD of three sites (we call this a {\em trichromatic VD}).
Such a VD has at most one trichromatic face (in addition to face $f$ itself).

\medskip
\noindent
{\bf Using point-location for distance oracles.}
To see why point-location on Voronoi diagrams is useful for distance oracles, we describe here the $\tilde O(n^{4/3})$-space $\tilde O(1)$-query oracle of \cite{CharalampopoulosGMW19} that we will be using.
It begins with an $r$-{\em division} of the graph $G$.
This is a partition of $G$ into $O(n/r)$ subgraphs (called {\em pieces}) such that every piece $\X$ contains $O(r)$ vertices and $O(\sqrt{r})$ {\em boundary} vertices $\partial \X$ (vertices shared by more than one piece).
An $r$-division can be computed in $\tilde O(n)$ time \cite{KleinMS13} with the additional property that the boundary vertices $\partial \X$ of every piece $\X$ lie on a constant number of faces of the piece (called {\em holes}).
To simplify the presentation, we assume that $\partial \X$ lies on a single hole (in general, we apply the same reasoning to each hole separately, and return the minimum distance found among the $O(1)$ holes).
The oracle consists of the following for each piece $\X$ of the $r$-division:
\begin{enumerate}
\item The $O(|\X|^{3/2})$ space, $O(|\X|^{3/2})$ construction time, $\tilde O(1)$ query-time distance oracle of~\cite{GawrychowskiMWW18} on the graph $\X$.
In total, these require $O(n\sqrt{r})$ space and construction time.
\item Two MSSP data structures~\cite{Klein02}, one for $\X$ and one for $\out{\X}=G  \setminus  (\X  \setminus  \partial \X)$, both with sources $\partial \X$.
The MSSP for $\X$ requires space $O(r \log r)$, and the MSSP for $\out{\X}$ requires space $O(n \log n)$.
Using these MSSPs, we can then query in $\tilde O(1)$ time the $u$-to-$v$ distance for any $u\in \partial \X$ and $v\in G$.
The total space and construction time of these MSSPs is $\tilde O({n^2}/{r})$, since there are $O(n/r)$ pieces.
\item For each vertex $u$ of $\X$, compute the Voronoi diagram $\VDin(u,\X)$ (resp. $\VDout(u,\X)$) for $\X$ (resp. $\out{\X}$) with sites $\partial \X$ and additive weights the distances from $u$ to these vertices in $G$ (the additive weights are computed in $\tilde O(|\partial \X|) = \tilde O(\sqrt{r})$ time using \cite{FakcharoenpholR06}).
Each Voronoi diagram can be represented in $O(\sqrt{r})$ space~\cite{GawrychowskiMWW18} and can be constructed in $\tilde O(\sqrt{n \sqrt{r}})$ time \cite{CharalampopoulosGMW19}.
Hence, all Voronoi diagrams require $O(n \sqrt{r})$ space and $\tilde O(n^{3/2} r^{1/4})$ construction time.
\end{enumerate}

To query a $u$-to-$v$ distance, let $\X$ be the piece that contains $u$.
If $v \not\in \X$ then the $u$-to-$v$ path must cross $\partial \X$.
We perform a point-location query for $v$ in  $\VDout(u,\X)$ in time $\tilde O(1)$ \cite{GawrychowskiMWW18}.
If $v\in \Vor(s)$ then we return the $u$-to-$s$ distance (from the precomputed additive weight) plus the $s$-to-$v$ distance (from the MSSP of $\out{\X}$).
Otherwise, $v \in \X$ and we return the minimum of two options: (1) The shortest $u$-to-$v$ path crosses $\partial \X$.
This is similar to the previous case except that the point-location query for $v$ is done in $\VDin(u,\X)$.
(2)
The shortest $u$-to-$v$ path does not cross $\partial \X$ (i.e., the path lies entirely within $\X$).
We  retrieve this distance by querying the distance oracle stored for $\X$.

By choosing $r=n^{2/3}$ we get an oracle of space $\tilde O(n^{4/3})$ and query-time $\tilde O(1)$.
The construction time however is $\tilde O(n^{5/3})$.
Notice that the only bottleneck preventing $\tilde O(n^{4/3})$ construction time is the construction of the Voronoi diagrams.

\medskip
\noindent
{\bf Shortest path separators.}
A {\em shortest path separator} $Q$ is a balanced cycle separator consisting of two shortest paths $P$ and $P'$ (emanating at the same vertex) plus a single edge.
A complete recursive decomposition tree $T$ of $G$ using shortest path separators can be obtained in linear time \cite{LT79}.

An arc $e=uv$ emanates (enters) left of a simple path $P$ if there exist two arcs $e_1$, $e_2$ in $P$, such that $u$ ($v$) is the head of $e_1$ and the tail of $e_2$, and $e$ appears between $e_1$ and $e_2$ in the clockwise order of arcs incident to $u$.
Otherwise, $e$  emanates (enters) right of $P$.
Notice that this is not defined for the endpoints of $P$, however, in the context in which we will use it, we will always extend $P$ so that its original endpoints will become internal.
Specifically, consider a shortest path separator $Q$ composed of an $x$-to-$y$ path $P$, an $x$-to-$y'$ path $P'$, and an edge $(y,y')$.
Then, we will extend $P$ on one side with $y'$ and on the other side with the vertex following $x$ on $P'$.
A symmetric extension will be done for $P'$.

\section{Efficient Computation of a Relaxed Partition}\label{sec:relaxed-partition}
In this section we define a relaxed partition, and introduce an algorithm that computes it in $\Otild(1)$ time.

\subsection{Definitions and data structures}

In our settings, let $F$ denote the face, and let $H\subseteq F$ denote a subset of vertices (sites) of $F$ (for the most part, we will use $|H|=3$ in order to find trichromatic faces, see \cref{sec:partition_to_trichromatic}).
We say that vertex $v\in P$ {\em prefers} a site $c \in H$ if $v$ belongs to $\Vor_H(c)$.
It is natural to partition the vertices of $P$ according to the site in $H$ that they prefer.
In the case of the alignment graph of~\cite{CharalampopoulosGM21}, each part in this partition is a contiguous interval of $P$.
In general planar graphs however, this is not the case, and the parts can be very fragmented (see \cref{fig:fragmented}).

\begin{figure}[htb]
\begin{minipage}[c]{0.25\textwidth}
    \includegraphics[width=\textwidth]{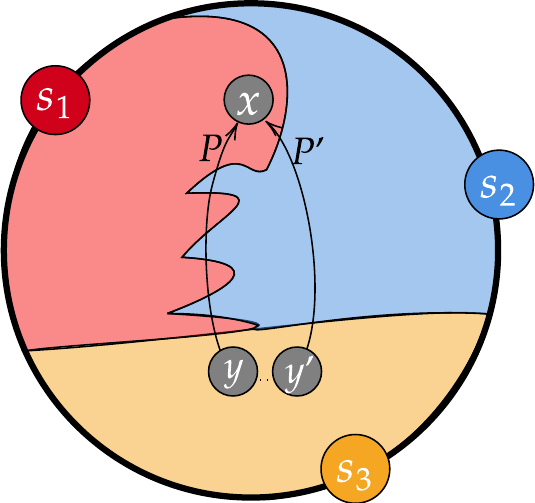}
  \end{minipage}\hfill
  \begin{minipage}[t]{0.7\textwidth}
    \caption{The partition of $P$ into parts corresponding to the Voronoi cells of $\VD(s_1,s_2,s_3)$ is very fragmented. \label{fig:fragmented}}
  \end{minipage}
\end{figure}

We observe that the partition would induce contiguous intervals if we only consider shortest paths that are allowed to enter $P$ from the right (or only from the left).
In that case we can say that a vertex $v \in P$ right-prefers (left-prefers) $c$.
We know how to determine efficiently whether  the true shortest path from a site $c$ to a vertex $v \in P$ enters $P$ from the right or from the left.
But the problem is that then we do not know how to efficiently report the shortest path from $c$ to $v$ that enters $P$ from the other side (because it is not the overall shortest path from $c$ to $v$).

To overcome this issue we observe that we can work with relaxed preferences, which we call \emph{like}.
If $v$ prefers $c$ (without restrictions on left/right), then we might as well say that $v$ both right-likes and left-likes $c$ (because either way this will point us to $c$).
Similarly, if $v$ prefers $c$ (again, without restrictions on left/right) then if the shortest path from $c$ to $v$ enters $P$ from the right, we do not really care what the left like of $v$ is (since in this case shortest paths entering $P$ from the left are just irrelevant), so we may as well say that $v$ left-likes $c'$ for any site $c'$.

Recall that the endpoints of $P$ are $x$ and $y$.
We denote by $s_x$ and $s_y$ the sites of $H$ such that $x\in \Vor_H(s_x)$  and $y\in\Vor_H(s_y)$.
In practice we will use the left-like (and symmetrically the right-like) relation with respect to a subset $S \subseteq H$ of sites.
If the true site in $S\cup\{s_x,s_y\}$ that $v$ prefers (without restrictions on left/right) is not one of the sites in $S$ then we may as well say that $v$ both left-likes and right-likes all sites in $S$ because eventually they will have no effect on the true answer.

The following definition formalizes the above discussion.
\begin{definition}[$(S,\leftside)$-like]
Let $P$ be an  $x$-to-$y$ shortest path, let
$H=\{h_1,h_2,\ldots h_{|H|}\}$ be a cyclic subsequence of a face $F$, and let $S=\{s_1,s_2,\ldots, s_{|S|}\}$ be a subsequence of $H$.
Let $s_x\in H$ (resp. $s_y$) be the site such that $x\in\Vor_H(s_x)$ (resp. $y\in \Vor_H(s_y)$).
A vertex $v\in P$ that belongs to the Voronoi cell $\Vor_{S\cup \{s_x,s_y\}}(c)$ is said to $(S,\leftside)$-like $s_i$ if at least one of the following holds: (1) $c=s_i$,
    (2) the shortest path from $c$ to $v$ enters $P$ from the right, or (3)  $c\notin S$.
\end{definition}

This notion of $(S,\leftside)$-like,  induces a {\em  relaxed partition} of $P$, in which the parts do form contiguous intervals along $P$.
 The relaxed partition has the property that for every $v\in P$  if $v\in \Vor_H(c)$ \textbf{and} $c$ reaches $P$ from some $direction$ (left or right) then in the $(H,direction)$-relaxed partition $v$ is in the part of $P$ corresponding to site $c$.
 In other words, together, the $(H,\leftside)$ and $(H,\rightside)$-relaxed partitions assign to each vertex of $P$ at most two sites in $H$, and it is guaranteed that one of these two sites is the true site of $v$.

\begin{definition}[Relaxed Partition]
Let $P$ be a shortest path, and let
$S$ and $H$ be cyclic subsequences of a face $F$ with $S\subseteq H$, denoted by $S=\{s_1,s_2,\ldots, s_{|S|}\}$ and $H=\{h_1,h_2,\ldots h_{|H|}\}$.

A relaxed $(S,\leftside)$-partition of $P$ w.r.t. $S$ is a partition of $P$ into $|S|$ disjoint subpaths $P_i=P[u_i,v_i]$ such that every $x\in P_i$ is a vertex that $(S,\leftside)$-likes $s_i$.
\end{definition}
The definitions of $(S,\rightside)$-like and of an $(S,\rightside)$-relaxed partition are symmetric.
The main part of our algorithm is dedicated to finding a relaxed partition.
First, in \cref{sec:ns-partition}
we will show how to find a relaxed partition of $P$ w.r.t. $S \subset H$ consisting of only two sites.
Then, in \cref{sec:H-partition}, we explain how to obtain a relaxed partition of $P$ w.r.t. $H$ by combining the relaxed partitions of pairs of sites in $H$.

For two vertices $u,v$ we denote the shortest path in $G$ between $u$ and $v$ as $R_{u,v}$.
For a site $s\in F$ we denote by $\Left(P,s)$ the set of vertices $v\in P$ such that $R_{s,v}$ enters $P$ from the left.

The following lemma describes a data structure that will later (\cref{lem:out_partition}) allow to find an $(F,\leftside)$-partition of $P$. A similar proof finds an $(F,\rightside)$-partition.

\begin{lemma}[Enhanced MSSP data structure]\label{lem:enhanced_mssp}
    Given a planar graph $G$, a face $f$ and a recursive decomposition $T$ of $G$, one can construct in $\Otild(n)$ time a data structure supporting the following queries, each in $\Otild(1)$ time. For any source $s\in f$, shortest path $P \in T$, and vertex $v\in P$:
    \begin{enumerate}
        \item $\dist(s,v)$ - return the distance between $s$ and $v$ in $G$.
        \item $\direction(s,v,P)$ - return whether $R_{s,v}$ enters $P$ from left or right.
        \item $\countA(s,P,i,j)$ for $i,j\in P$ - return the number of vertices $x$ on $P[i,j]\cap \Left(P,s)$.
        \item $\select(s,P,i,j,k)$ for $i,j\in P$ - return the $k$'th vertex $x$ of $P[i,j]\cap \Left(P,s)$.
        \item $\ancestor(s,v,d)$ - return the ancestor of $v$ that is in depth $d$ in the shortest paths tree of $s$.
    \end{enumerate}
\end{lemma}

The proof of  \cref{lem:enhanced_mssp} is deferred to \cref{sec:enhanced_mssp}.
In the rest of this section we prove the following lemma.

\begin{lemma}\label{lem:out_partition}
Given  $F$, $H$, and $T$, together with their enhanced MSSP, and a shortest path  $P \in T$, one can compute an $(H,\leftside)$-partition of $P$, in $\Otild(|H|)=\Otild(1)$ time.
\end{lemma}

\subsection{Setup of the query}

Recall that the endpoints of $P$ are $x$ and $y$, and that $s_x$ and $s_y$ are the sites of $H$ such that $x\in \Vor_H(s_x)$  and $y\in\Vor_H(s_y)$.
Notice that it is straightforward to find $s_x$ and $s_y$ in $\Otild(|H|)=\Otild(1)$ time by explicitly checking the distance from each of the sites in $H$ to $x$  and to $y$ using the enhanced MSSP.
We denote $R_x=R_{s_x,x}$ and $R_y=R_{s_y,y}$.
We assume $s_x\ne s_y$, since the case $s_x=s_y$ is a degenerated case.\footnote{
If $s_x=s_y$ one can think of the graph obtained by making an incision along $R_x$, which allows us to think of $s_x$ as being split into two vertices, where one is the site of $x$ and the other is the site of $y$, see also \cref{remark:linear_time_partition}.}
We also assume that $H\cap P=\emptyset$, since otherwise we find a partition for $H\setminus P$, and then add the intervals of the sites in $H\cap P$, using a binary search along $P$.
The sites $s_x$ and $s_y$ partition the face $F$ into two intervals which we denote by $F_{\leftside}$ and $F_{right}$.
In addition, let $H_{\leftside}=F_{\leftside}\cap H$ and $H_{right}=F_{right}\cap H$.
The following lemma (see proof in \cref{app:missing_proof}) states that when computing the
$(H,\leftside)$-partition, we can ignore every $s\in H_{right}$.
Thus, every time we consider a partition with respect to some $S\subseteq H_{\leftside}$, we consider the vertices of $S$ from the closest vertex to $s_x$ to the closest vertex to $s_y$ (on $F_{\leftside}$).
In addition, the partition is always such that the order of the parts corresponding to vertices of $S$ is from the one containing $x$ to the one containing $y$.

\begin{lemma}\label{lem:Hleft-is-enough}
Every $(H_{\leftside},\leftside)$-partition is an $(H,\leftside)$-partition.
\end{lemma}

\begin{remark}\label{remark:order}
When considering a partition of $P$ with respect to a set $S\subseteq H_{\leftside}$, the parts of $P$ in our partition are chosen so that they are consistent with the order of the vertices in $S$.
For example, the first part, corresponds to the vertex of $S$ closest to $s_x$,  contains $x$ (unless it is empty) and the part of the vertex closest to $s_y$, contains $y$ (unless it is empty).
\end{remark}

\noindent
{\bf Trimming $P$.}
By definition $x\in \Vor_H(s_x)$ and for every $v\in R_x$ we also have $v\in \Vor_H(s_x)$.
Therefore, w.l.o.g., we assume that $R_x\cap P$ contains only the vertex $x$, since all the vertices of $R_x\cap P$ are known to belong to $\Vor_H(s_x)$.
Similarly, we assume that $R_y\cap P= \{y\}$.

\medskip
\noindent
{\bf Swirly paths.}
For vertices $s\in F_{\leftside}$ and $v\in P$, we say that $R_{s,v}$ is a {\em swirly} path if $v\in\Left(P,s)$ and $R_{s,v}$ crosses $R_x$ or $R_y$.
Moreover, a path is called $x$-swirly (resp. $y$-swirly) if the first path it crosses among $R_x$ and $R_y$ is $R_x$ (resp. $R_y$). See \cref{fig:setup}.

Swirly paths complicate life.
We shall first (in \cref{sec:ns-partition}) explain how to obtain an $(S,\leftside)$-partition $P$ when $|S|=2$ and there are no swirly paths.
Next, in~\cref{sec:xy-partition} we show how to obtain an $(\{s_x,s_y\},\leftside)$ partition.
Then, in \cref{sec:two-partition}, we will show how to use the partition for $\{s_x,s_y\}$ to obtain  the partition for any $|S|=2$, even if there are swirly paths.
Finally, in \cref{sec:H-partition} we will show how to obtain an $(H,\leftside)$-partition of $P$, removing the restriction to two vertices and proving \cref{lem:out_partition}.

\begin{figure}[htb]
\begin{minipage}[c]{0.25\textwidth}
    \includegraphics[width=\textwidth]{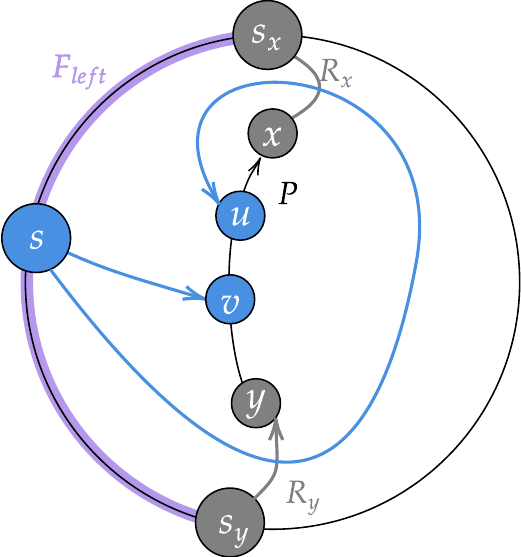}
  \end{minipage}\hfill
  \begin{minipage}[t]{0.7\textwidth}
    \caption{The path $P$ with endpoints $x$ and $y$. We think of $P$ as being oriented from $y$ to $x$. The corresponding sites $s_x$ and $s_y$ with $R_x=R_{s_x,x}$ and $R_y=R_{s_y,y}$.
   The site $s$ is on $F_{\leftside}$, the path $R_{s,v}$ is a non-swirly path, and the path $R_{s,u}$ is a $y$-swirly path. \label{fig:setup}}
  \end{minipage}
\end{figure}

\subsection{Partition with respect to two sites}\label{sec:ns-partition}

In this subsection we show how to compute an $(S,\leftside)$-partition when $|S|=2$.
We start with a special case, where we consider a subpath $\hat P$ of $P$ that is not involved in any swirly paths.
Later, in \cref{lem:s1s2partition} we introduce the algorithm for computing an $(S,\leftside)$-partition when $|S|=2$.

\begin{lemma}\label{lem:out_partition_non-swirly}
Let $s_1,s_2\in F_{\leftside}\setminus P$ be two sites such that $s_1$ is closer to $s_x$ on $F_{\leftside}$ than $s_2$.
Let $\hat P=P[a,b]$ be a subpath of $P$ with $a$ being closer to $x$ than $b$, such that for every $i\in\{1,2\}$ and $v\in \Left(\hat P,s_i)$ the path $R_{s_i,v}$ is non-swirly.
One can compute in $\Otild(1)$ time an $(\{s_1,s_2\},\leftside)$-partition of $\hat P$ into $\hat P_1,\hat P_2$ such that $a\in \hat P_1$ and $b\in\hat P_2$, unless the partition is trivial.
\end{lemma}

The proof of \cref{lem:out_partition_non-swirly} consists of several steps.
We first introduce the notion of \emph{winning}, and show (\cref{clm:winner_takes_something}) that it is a monotone property of vertices of $\hat P$.
We will then use this monotonicity to perform binary search.
For a vertex $v\in \hat P$ we say that $s_1$ {\em wins at $v$} if there exists an $s_1$-to-$v$ non-swirly path $R$ that enters $P$ from the left and $\len(R)<\dist(s_2,v)$. We define $s_2$ winning at $v$ symmetrically.

To avoid clutter we assume $\hat P=P$, otherwise one just has to change in the following proofs $P$ to $\hat P$ and $x,y$ to the endpoints of $\hat P$.

\begin{figure}[htb]
     \centering
     \includegraphics[width=0.3\textwidth]{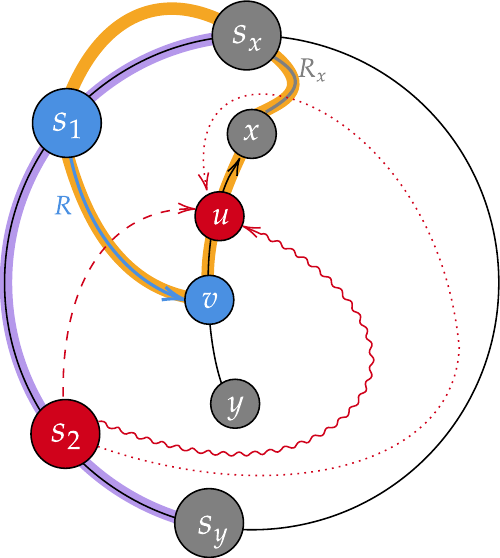}
     \hspace{1in}\includegraphics[width=0.31\textwidth]{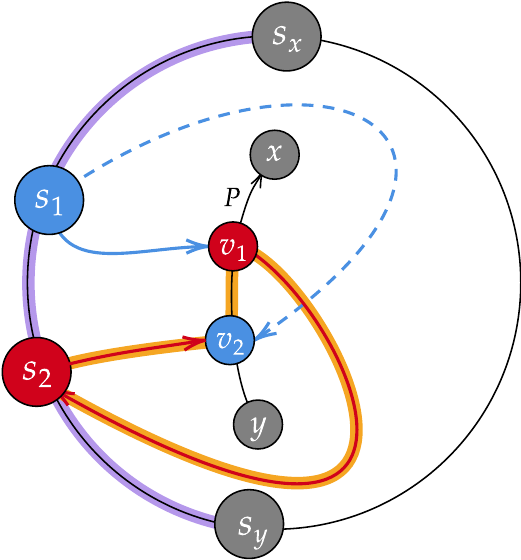}

     \caption{Left: The (orange) cycle $C$ in the proof of \cref{clm:winner_takes_something}. The three options for $R_{s_2,u}$ are in red. Dashed red contradicts $v$ being closer to $s_1$ than to $s_2$, dotted red contradicts $x$ being closer to $s_x$ than to $s_2$, and wavy red contradicts $R_{s_2,u}$ entering $P$ from the left. Right: The (orange) cycle $C$ in the proof of \cref{clm:no-criss-cross}. Since $s_1$ is (strictly) on one side of $C$ and $v_2$ is on the other side of $C$,  the dashed  $R_{s_1,v_2}$ path must cross $C$ in either  $P[v_2,v_1)$ or  $R_{s_2,v_2}$.
        Both these crosses are impossible.}
\label{fig:winner_takes_something}
 \end{figure}
\begin{claim} \label{clm:winner_takes_something}
Let $v\in P$ such that $s_1$ wins at $v$.
Then, every $u\in P[x,v]$ is a vertex that $(\{s_1,s_2\},\leftside)$-likes $s_1$.
Symmetrically, if $s_2$ wins at $v$, then every $u\in P[v,y]$ is a vertex that $(\{s_1,s_2\},\leftside)$-likes $s_2$.
\end{claim}
\begin{proof}
We prove the first claim, the proof of the symmetric claim is similar.
Since $s_1$ wins at $v$, there exists an $s_1$-to-$v$ non-swirly path $R$ that enters $P$ from the left and $\len(R)<\dist(s_2,v)$.
Let $u\in P[x,v]$.
Assume by contradiction that $u$ is a vertex that does not  $(\{s_1,s_2\},\leftside)$-like $s_1$.
Then $R_{s_2,u}$ must enter $P$ from the left and $\len(R_{s_2,u})<\min \{\dist(s_1,u),\dist(s_x,u)\}$.
Hence, for any $z\in R_{s_2,u}$ we have $\dist(s_2,z)<\dist(s_1,z)$.
Thus, $R\cap R_{s_2,u}=\emptyset$.
Similarly,  for any $z\in R_x$ we have $\dist(s_x,z)<\dist(s_2,z)$.
Hence, $R_x\cap R_{s_2,u}=\emptyset$.

Consider the cycle $C$ composed of $R$,$ P[v,x]$, $ R_x$ and the  Jordan curve connecting $s_x$ and $s_1$ embedded in the face $F$. See \cref{fig:winner_takes_something} (left).
Notice that $C$ has no self-crossing since $R$ does not cross $ R_x$ as $R$ is a non-swirly path.
We think of $C$ as an oriented cycle whose orientation is consistent with that of $P$ (recall that the orientation of $P$ is from $y$ to $x$) so as to define left and right properly (i.e., define what is left and right w.r.t. $C$).
Observe that $s_2$ is on the right side of $C$.
Since $R_{s_2,u}$ intersects $C$, but does not intersect $R$ nor $R_x$, it must be that $R_{s_2,u}$ intersects only $P$ among the parts of $C$.
Since $s_2$ is on the right side of $C$, $R_{s_2,u}$ enters $P$ from the right, a contradiction.
\end{proof}

Having proved that the winning property  is sufficient to perform a binary search to obtain a partition of $P$, we next show how to find a pair $(s_w,v_t)$ (where $w\in\{1,2\}$ and $t\in\{1,2\}$) such that $s_w$ is a winner at $v_t$.
We distinguish two scenarios that are handled differently (\cref{clm:no-criss-cross,clm:criss-cross}).
Later we will put it all together and show how to make this into an efficient binary search procedure.

    \begin{claim}\label{clm:no-criss-cross}
        Let $v_1$ and $v_2$ be  vertices such that $v_1\in\Left(P,s_1)$ and $v_2\in\Left(P,s_2)$.
        If $v_1$ is closer on $P$ to $x$ than $v_2$ then there is an $\Otild(1)$-time algorithm that outputs $w\in\{1,2\}$ and $t\in\{1,2\}$ such that $s_w$ wins at $v_t$.
    \end{claim}
    \begin{proof}
        For every $w,t$, we check if $R_{s_w,v_t}$ enters $P$ from the left and whether $\dist(s_w,v_t)$ is smaller than the distance of the other site to $v_t$.
        We return a pair $w,t$ satisfying the claim.
        This takes $\Otild(1)$ time using the enhanced MSSP data structure of \cref{lem:enhanced_mssp}.

        It remains to prove that at least one such pair exists.
        Assume by contradiction that no such pair exists.
        In particular, since $R_{s_1,v_1}$ enters $P$ from the left, it must be the case that $\dist(s_2,v_1)<\dist(s_1,v_1)$.
        Moreover, $R_{s_2,v_1}$ must enter $P$ from the right (otherwise $s_2,v_1$ is a valid pair).
        Similarly, since $R_{s_2,v_2}$ enters $P$ from the left we have $\dist(s_1,v_2)<\dist(s_2,v_2)$.

        We will show that $R_{s_1,v_2}$ enters $P$ from the left, implying that $(s_1,v_2)$ is a valid pair, contradicting our assumption.
        Consider the (non self crossing) cycle $C =R_{s_2,v_2}\circ P[v_2,v_1]\circ R_{s_2,v_1}$.
        We think of $C$ as an oriented cycle whose orientation is consistent with that of $P$ so as to define left and right properly, see \cref{fig:winner_takes_something} (right).

        Let $u$ be the first vertex of $R_{s_1,v_2}$ that belongs to $C$.
        Note that $u$ exists since $v_2 \in C$.
        Since $s_1$ is on the left side of $C$, $R_{s_1,v_2}$ enters $C$ from the left at $u$.
        We will show that $u\notin R_{s_2,v_1}$, which means that either (i) $u\in P(v_1,v_2]$ or (ii) $u \in R_{s_2,v_2}$.
        To see that  $u\notin R_{s_2,v_1}$ notice that every $z\in R_{s_2,v_1}$ satisfies $\dist(s_2,z)<\dist(s_1,z)$ and every $z\in R_{s_1,v_2}$ satisfies $\dist(s_1,z)<\dist(s_2,z)$.
        If (i) $u \in P(v_1,v_2]$  then $R_{s_2,v_2}$ enters $P$ from the left, as it enters $C$ from the left at $u$.
        If (ii) $u \in R_{s_2,v_2}$ then $R_{s_1,v_2}[u,v_2]=R_{s_2,v_2}[u,v_2]$, meaning that $R_{s_1,v_2}$ enters $P$ from the left.
        To conclude, in both cases $R_{s_1,v_2}$ enters $P$ from the left, which contradicts our assumption.
        \qedhere

    \end{proof}

\begin{figure}[ht]
    \centering
    \begin{subfigure}[t]{0.3\textwidth}
        \centering
        \includegraphics[width=\textwidth]{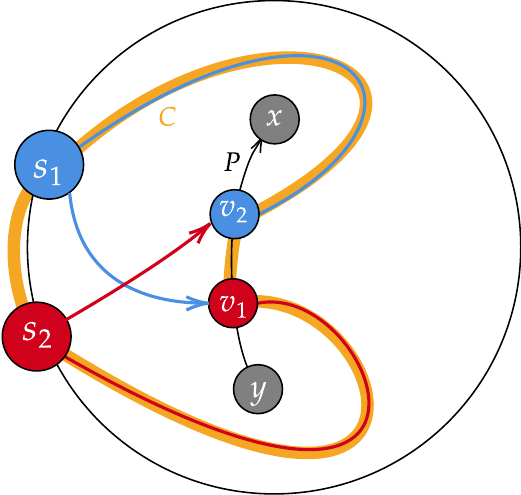}
               \label{fig:figure1}
    \end{subfigure}
    \hspace{1in}
    \begin{subfigure}[t]{0.3\textwidth}
        \centering
        \includegraphics[width=\textwidth]{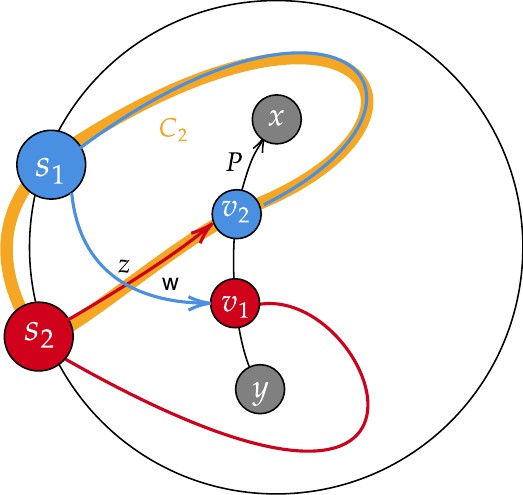}

        \label{fig:figure2}
    \end{subfigure}
    \caption{The (orange) cycles $C$ (left image) and $C_2$ (right image) in the proof of \cref{clm:criss-cross}.
    The situation here is that $s_1$ reaches $v_1$ from the left and loses to $s_2$ that reaches $v_1$ from the right, and the symmetric issue for $v_2$.
    Thus, $v_1$ and $v_2$ both $(\{s_1,s_2\},\leftside)$-likes both $s_1$ and $s_2$.
    To recognize a winner we detect which site is closer to $z$.}
    \label{fig:criss-cross}
\end{figure}

\begin{claim}\label{clm:criss-cross}
    Let $v_1$ and $v_2$ be  vertices such that $v_1\in\Left(P,s_1)$ and $v_2\in\Left(P,s_2)$.
    If $v_2$ is closer on $P$ to $x$ than $v_1$ then there is an $\Otild(1)$-time algorithm that outputs $w\in\{1,2\}$ and $t\in\{1,2\}$ such that $s_w$ wins at $v_t$.
\end{claim}
\begin{proof}
For every $w,t$ we check if $R_{s_w,v_t}$ enters $P$ from the left and whether $\dist(s_w,v_t)$ is smaller than the distance of the other site to $v_t$.
If such a pair is found, we return  this pair.
Otherwise, it must be that: (1) $\dist(s_2,v_1)<\dist(s_1,v_1)$, (2) $R_{s_2,v_1}$ enters $P$ from the right, (3) $\dist(s_1,v_2)<\dist(s_2,v_2)$, and (4) $R_{s_1,v_2}$ enters $P$ from the right.

Consider the cycle $C$ composed of $R_{s_2,v_1}$, $P[v_1,v_2]$  $R_{s_1,v_2}$, and the Jordan curve connecting $s_1$ and $s_2$ embedded in the face $F$.
We first claim that all edges entering $P$ from the left are on the left side of $C$.
To see this, consider making an incision along $P$ (duplicating its vertices and edges).
Since both $R_{s_1,v_2}$ and $R_{s_2,v_1}$ do not cross $P$ and enter $P$ from the right, they both use the right copy of $P$, and do not intersect the left copy of $P$. Hence, the left copy of $P$, and hence also all edges entering $P$ from the left, are on the left side of $C$.

The path $R_{s_1,v_1}$  does not intersect $R_{s_2,v_1}\setminus P$ (since the former enters $P$ from the left and the latter enters $P$ from the right).
Moreover, by uniqueness of shortest paths, $R_{s_1,v_1}$  does not cross $R_{s_1,v_2}$ nor $P$.
Thus, $R_{s_1,v_1}$ does not cross $C$.
Symmetrically, $R_{s_2,v_2}$ does not cross $C$.
Since both $R_{s_1,v_1}$  and  $R_{s_2,v_2}$ start on $C$, do not cross $C$ and enter $P$ from the left, they must be on the left side of $C$.

Since $\dist(s_1,v_2)<\dist(s_2,v_2)$, and $\dist(s_2,v_1)<\dist(s_1,v_1)$ it must be  that $R_{s_1,v_2}\cap R_{s_2,v_1}=\emptyset$.
Hence, in the counterclockwise cyclic order of $C$ the vertices appear as $s_2,v_1,v_2,s_1$.
Therefore, the paths $R_{s_1,v_1}$  and  $R_{s_2,v_2}$  form a cross configuration and must cross each other, such that $R_{s_1,v_1}$ crosses $R_{s_2,v_2}$ from left to right.

For simplicity, we assume that $R_{s_1,v_1}\cap R_{s_2,v_2}$ contains a single vertex $z$.
In the general case, it may be a contiguous subpath and the proof is similar.
Our goal now is to determine whether $z$ is closer to $s_1$ or to $s_2$.
We emphasize that the algorithm does not find $z$ itself at any time.

Consider the path $R_{s_1,v_1}$ and notice that $\dist(s_1,s_1)<\dist(s_2,s_1)$ and $\dist(s_2,v_1)<\dist(s_1,v_1)$.
Since $R_{s_1,v_1}$ is a shortest path, it must be that for some  vertex $w\in R_{s_1,v_1}$  we have that all vertices in $R_{s_1,v_1}[s_1,w)$ are closer to $s_1$ than to $s_2$ and all vertices in $R_{s_1,v_1}[w,v_1]$ are closer to $s_2$ than to $s_1$.
Using $\ancestor(s_1,v_1,d)$ queries of
the enhanced MSSP (in $\Otild(1)$ time per query), the algorithm binary-searches for $w$, checking at every step which of $s_1$ and $s_2$ is closer to the queried vertex.
After finding $w$, the algorithm uses enhanced MSSP
to check whether $w$ appears to the right, to the left, on $R_{s_2,v_2}$
in the shortest path tree rooted at $s_2$.
We note that $w$ cannot be a descendant of $v_2$ in the tree of $s_2$, since every vertex on $R_{s_2,w}$ is closer to $s_2$ than to $s_1$, while $v_2$ is closer to $s_1$ than to $s_2$.
If $w$ is on $R_{s_2,v_2}$ then $z=w$ is closer to $s_2$ than to $s_1$.

We will show that if $w$ is to the right of $R_{s_2,v_2}$ in the tree of $s_2$ then $z$ is closer to $s_1$ than to $s_2$.
Consider the cycle $C_2$ obtained by concatenating $R_{s_2,v_2}, R_{v_2,s_1}$ and the Jordan curve connecting $s_1$ and $s_2$ embedded in the face $F$.
We consider $C_2$ to be oriented from $s_2$ to $v_2$.
Notice that $R_{s_1,v_1}[s_1,z]$ does not cross $C_2$ since $R_{s_1,v_1}$ and $R_{s_2,v_2}$ cross each other only once.
Recall that $R_{s_1,v_1}[s_1,z]$ is to the left of $R_{s_2,v_2}$, and therefore to the left of $C_2$.
Additionally recall $R_{s_1,v_1}[z,v_1]$ is to the right of $R_{s_2,v_2}$, and therefore to the right of $C_2$.
Consider the first vertex $u$ of $R_{s_2,w}$ which is not on $R_{s_2,v_2}$.
Since $w$ is to the right of $R_{s_2,v_2}$, it holds that $u$ is to the right of $C_2$.
Moreover $Q=R_{s_2,w}[u,w]$ is strictly to the right of $C_2$.
This is because $Q$ is disjoint from $R_{s_2,v_2}$ by uniqueness of shortest paths, and is disjoint from $R_{s_1,v_2}$ since $v_2$ is closer to $s_1$ and $w$ is closer to $s_2$.
Thus, $Q$ is disjoint from $C_2$, and therefore $w$ is to the right of $C_2$ and therefore $w\in R_{s_1,v_1}(z,v_1]$ and this means $z$ is closer to $s_1$ than to $s_2$ by definition of $w$.
A similar argument shows that if $w$ is to the left of $R_{s_2,v_2}$ in the shortest paths tree of $s_2$, then $z$ is closer to $s_2$ than to $s_1$.

Thus, the algorithm deduces whether $z$ is closer to $s_1$ or to $s_2$.
If $z$ is closer to $s_1$, the algorithm reports that $s_1$ wins at $v_2$.
This is because the path $R=R_{s_1,v_1}[s_1,z]\circ R_{s_2,v_2}[z,v_2]$ enters $P$ from the left and $\len(R)<\dist(s_2,v_2)$.
Moreover, $R$ is a concatenation of subpaths of two non-swirly paths from vertices on $F_{\leftside}$, and therefore $R$ is a non-swirly path.
Similarly, if $z$ is closer to $s_2$, the algorithm reports that $s_2$ wins at $v_1$.
This is because the path $R=R_{s_2,v_2}[s_2,z]\circ R_{s_1,v_1}[z,v_1]$ enters $P$ from the left and $\len(R)<\dist(s_1,v_1)$, and $R$ is a non-swirly path.

Finally, the running time of the algorithm is indeed $\Otild(1)$, since all distance queries, and directions can be answered in $\Otild(1)$ time per query by the enhanced MSSP of \cref{lem:enhanced_mssp} and the binary search of $w$ increases the running time only by an additional $\Otild(1)$ factor.
\end{proof}

Using \cref{clm:winner_takes_something,clm:no-criss-cross,clm:criss-cross} we introduce a recursive binary search partition algorithm, completing the proof of \cref{lem:out_partition_non-swirly}.
In each recursive step, the algorithm works on a subpath $\bar P=P[a,b]$, and returns a partition of $\bar P$ into a (possibly empty) prefix $\bar P_1$ and a (possibly empty) suffix $\bar P_2$  such that every $v\in\bar P_i$  is a vertex that $(\{s_1,s_2\},\leftside)$-likes $s_i$ (for $i\in\{1,2\}$).

When working on $\bar P = P[a,b]$, the algorithm starts by finding a `left median' vertex of $s_1$ in $P[a,b]$.
Formally, a left median of $s_1$ in $P[a,b]$ is a vertex $v_1 \in P[a,b]$ such that the shortest path between $s_1$ and $v_1$ enters $P$ from the left, and the number of vertices whose shortest path from $s_1$ enters $P$ from the left  in $P[a,v_1]$ differs by at most  $1$ from the number of vertices whose shortest path from $s_1$ enters $P$ from the left in $P[v_1,b]$.
The algorithm finds left medians $v_1$ and $v_2$ for $s_1$ and $s_2$, respectively, in $P[a,b]$ by using a $\countA$ query to the enhanced MSSP, and then applying binary search using $\select$ and $\countA$ queries.

If no left-median exists for $s_1$, the algorithm returns the partition $\bar P_1 = \emptyset$ and $\bar P_2 = \bar P = P[a,b]$.
Similarly, if no left-median exists for $s_2$, the algorithm returns the partition $\bar P_1 = \bar{P}$ and $\bar P_2 = \emptyset$.
Otherwise, depending on the order of $v_1$ and $v_2$ on $P$, the algorithm either applies \cref{clm:no-criss-cross} or \cref{clm:criss-cross} to find $w\in \{ 1,2\}$ and $t\in \{1,2\}$ such that $s_w$ wins at $v_t$.
If $w=1$, the algorithm recursively obtains a partition $(P'_1,P'_2)$ of $P(v_t,b]$, and returns $\bar P_1 = P[a,v_t] \circ P'_1$ and $\bar P_2 = P'_2$.
If $w = 2$, the algorithm recursively obtains a partition $(P'_1,P'_2)$ of $P[a,v_t)$ and returns $\bar P_1 = P'_1$ and $\bar P_2 = P'_2 \circ P[v_t,b]$.

\medskip
\noindent
{\bf Correctness.}
The correctness of the halting condition follows from the definition of an $(\{s_1,s_2\},\leftside)$-partition.
If there is no left-median for $s_1$ (resp. $s_2$), in particular there are no vertices that $s_1$ (resp. $s_2$) reaches from the left on $\bar P$.
Therefore, every vertex on $\bar P$ is a vertex that $(\{s_1,s_2\},\leftside)$-likes $s_2$ (resp. $s_1$) and therefore a partition that sets $P_2 = \bar P$ (resp. $P_1 = \bar{P}$) is valid.
The correctness of the recursive step of the algorithm follows directly from \cref{clm:winner_takes_something} and the correctness of the recursion.

\medskip
\noindent
{\bf Complexity.}
The non-recursive part of the algorithm consists of a polylogarithmic number of queries for the enhanced MSSP, and therefore takes $\Otild(1)$ time.
We claim that the recursion depth is bounded by $O(\log n)$.
This is because in every recursive step the algorithm either reduces the number of vertices that $s_1$ reaches from the left or the number of vertices that $s_2$ reaches from the left - by half.
Since initially each of these numbers is bounded by $|P| \le n$, a halting condition must be satisfied after at most $2 \log n$ recursive calls.
The overall time complexity is therefore $\Otild(1)$ thus completing the proof of \cref{lem:out_partition_non-swirly}.\qed

\medskip
The following Lemma is the general version of \cref{lem:out_partition_non-swirly} in the presence of swirly paths. Its proof is in \cref{app:two_sites}.

\begin{restatable}{lemma}{lemtwositespartition} \label{lem:s1s2partition}
Let $s_1,s_2\in F_{\leftside}\setminus P$ be two sites such that $s_1$ is closer to $s_x$ on $F_{\leftside}$ than $s_2$.
One can compute in $\Otild(1)$ time an $(\{s_1,s_2\},\leftside)$-partition of $P$ into $P_1,P_2$ such that $x\in  P_1$ and $y\in P_2$, unless the partition is trivial.
\end{restatable}

\subsection{Relaxed partition for $H$ (proof of \cref{lem:out_partition})}\label{sec:H-partition}

Using \cref{lem:s1s2partition} we are finally ready to prove \cref{lem:out_partition} on the construction of a relaxed partition for $H$.
By \cref{lem:Hleft-is-enough}, it is enough to compute an $(H_{\leftside},\leftside)$-partition.
If $|H_{\leftside}|=1$ the partition is trivial. If $|H_{\leftside}|=2$ we obtain the partition from \cref{lem:s1s2partition}.
When $|H_{\leftside}|=3$ with $s_1,s_2,s_3$ being the three sites according to their order on $F_{\leftside}$, starting from $s_1 = s_x$, we apply \cref{lem:s1s2partition} on three pairs $(s_1,s_2)$, $(s_2,s_3)$ and $(s_1,s_3)$.
Let $P_i^{i,j}$ denote that part corresponds to $s_i$ in the partition computed for $(s_i,s_j)$.
If $P^{1,2}_2\cap P^{2,3}_2\ne\emptyset$ the algorithm returns $P^{1,2}_1,P^{1,2}_2\cap P^{2,3}_2,P^{2,3}_2$.
Otherwise, if $P^{1,2}_2\cap P^{2,3}_2=\emptyset$  the algorithm returns $P^{1,3}_1,P^{1,3}_3$.
Clearly, the algorithm takes $\Otild(1)$ time.

\medskip
\noindent
{\bf Correctness.}
We first consider the case where $P^{1,2}_2\cap P^{2,3}_2=\emptyset$.
Let $v\in P^{1,3}_1$ and assume to the contrary that $v$ does not $(\{s_1,s_2,s_3\},\leftside)$-like $s_1$.
It must be that there exists $c\in\{s_1,s_2,s_3,s_x,s_y\}$ such that $v\in\Vor_{\{s_1,s_2,s_3,s_x,s_y\}}(c)$ and $v\in\Left(P,c)$.
since $v$ is a vertex that $(\{s_1,s_3\},\leftside)$-likes $s_1$, it must be that $c=s_2$.
However, this means that $v$ does not $(\{s_1,s_2\},\leftside)$-like $s_1$ and also $v$ does not $(\{s_2,s_3\},\leftside)$-like $s_3$.
Hence, $v\in P^{1,2}_2\cap P^{2,3}_2$, a contradiction.
Thus, $v$ indeed $(\{s_1,s_2,s_3\},\leftside)$-likes $s_1$.
The proof for vertices in $P^{1,3}_3$ is similar.

We next consider the case $P^{1,2}_2\cap P^{2,3}_2\ne \emptyset$.
For this case we first prove that every $v\in P_2=P^{1,2}_2\cap P^{2,3}_2$ must $(\{s_1,s_2,s_3\},\leftside)$-like $s_2$.
Assume to the contrary that there exists some $v\in P_2$ such that $v\in\Vor_{\{s_1,s_2,s_3,s_x,s_y\}}(c)$ and $v\in\Left(P,c)$ and $c\in\{s_1,s_2,s_3\}$ and $c\ne s_2$.
W.l.o.g. we can assume $c=s_1$ (the case $c=s_3$ is symmetric).
However, in this case $v$ does not $(\{s_1,s_2\},\leftside)$-like $s_2$, contradicting $v\in P^{1,2}_2$.

Finally, we prove that every $v\in P^{1,2}_1$ must $(\{s_1,s_2,s_3\},\leftside)$-like $s_1$ (the proof for $P^{2,3}_3$ is symmetric).
Notice that $P^{1,2}_1\subset P^{2,3}_2$ since $P^{1,2}_2\cap P^{2,3}_2\ne \emptyset$.
Assume to the contrary there exists some $v\in P^{1,2}_1$ such that $v\in\Vor_{\{s_1,s_2,s_3,s_x,s_y\}}(c)$ and $v\in\Left(P,c)$ and $c\in\{s_1,s_2,s_3\}$ and $c\ne s_1$.
It cannot be that $c=s_2$, because $v\in P^{1,2}_1$.
Assume $c=s_3$, then we have that $v$ does not $(\{s_2,s_3\},\leftside)$-like $s_2$, contradicting $v\in P^{1,2}_1\subset P^{2,3}_2$.

\medskip
\noindent
{\bf Complexity.} Finally, the running time of the algorithm is $\Otild(1)$ since all distance and direction queries are answered in $\Otild(1)$ time per query by the enhanced MSSP data structure of \cref{lem:enhanced_mssp}
and the binary search of $w$ increases the running time only by additional $\Otild(1)$ time.
This concludes the proof of \cref{lem:out_partition}.

\begin{remark}\label{remark:linear_time_partition}
    We note that the proof of \cref{lem:out_partition} is written for $|H|=3$.
    However, the algorithm described in the proof can be easily generalized to an  algorithm that constructs a partition for a set of $k$ sites from a partition of $k-1$ sites.
    A standard amortization argument shows that the total running time for constructing a partition of $k$ sites is $\Otild(k)$.
\end{remark}

\section{From a Relaxed Partition to a Trichromatic Face}\label{sec:partition_to_trichromatic}
In this section we prove the following lemma, stating that one can find a trichromatic face in $\Otild(1)$ time.

\begin{lemma}\label{lem:trichromatic}
    Given a planar graph $G$, and a face $F$, one can preprocess $G$ in $\Otild(|G|)$ time such that given a subset $H\subseteq F$ of 3 sites, and additive weights $w: H \rightarrow \R^+$, one can compute the trichromatic face of $\VD(H,w)$ in time $\Otild(1)$.
\end{lemma}
\begin{proof}
Recall that a trichromatic VD (a VD with three sites) has at most one trichromatic face $\hat f$ (apart from the face containing the three sites). We show how to identify $\hat f$ recursively, starting from the root $G$ of the recursive decomposition tree $T$. At a step of the recursion involving a subgraph $\X$, we determine whether $\hat f$ is enclosed by the cycle separator $Q$ of $\X$ or not, and recurse on the appropriate child of $\X$ in $T$ until we get to a subgraph of constant size that contains $\hat f$ (in which we identify $\hat f$ trivially using $O(1)$ distance queries from the sites using the MSSP data structure in $O(\log n)$ time).

We obtain a relaxed partition of $Q$ using \cref{lem:out_partition}. We next  explain how to use the relaxed partition to deduce whether $\hat f$ is enclosed by $Q$ or not.
The following immediate consequence of the Jordan curve theorem implies that it suffices to compute the parity of the number of crossings of the three bisectors forming the trichromatic VD.
The challenge is that each bisector may cross $Q$ many times (even $\Theta(n)$ times), so we cannot afford to actually count all of the crossings in order to compute the parity.

\begin{claim}
    The trichromatic face $\hat f$ is enclosed by $Q$ if and only if each of the 3 bisectors that meet at $\hat f$ crosses $Q$ an odd number of times.
\end{claim}
\begin{proof}
    In the trichromatic VD, each of the three bisectors originates at the face $F$ and terminates at $\hat f$. The cycle separator $Q$ partitions the faces of $G$ into two sets, and by definition of $F$ as the infinite face, and of enclosure, $F$ is not enclosed by $Q$. By the Jordan Curve theorem, any path in the plane starts at a point not enclosed by $Q$ and ends at a point enclosed by $Q$ if and only if it crosses $Q$ an odd number of times.
\end{proof}

The relaxed partition of $Q$ consists of $O(1)$ intervals for each of the two paths forming $Q$, and each of the two sides of these paths.
The $O(1)$ endpoints of all these intervals partition $Q$ into $O(1)$ intervals.
Consider such an interval $I$ of $Q$. The vertices of $I$ are all assigned by the relaxed partition a single site from each side. Hence, the vertices of $I$ belong to at most two Voronoi cells in VD. Therefore, the interval $I$ can only be crossed by the bisector between these two cells.
Moreover, the bisector crosses $I$ an even number of times if and only if both endpoints of  $I$ belong to the same cell. We can determine whether this is the case or not by using $O(1)$ MSSP queries on the endpoints of $I$.

Thus, we can determine the parity of the number of crossings of $Q$ by each bisector by summing the parities contributed by each of the $O(1)$ intervals, and deduce if $\hat f$ is enclosed by $Q$ in $O(\log n)$ time.
\end{proof}

\section{Computing the Voronoi Diagram}\label{sec:computingVD}

In this section, we describe an algorithm that, given access to the mechanism for computing trichromatic faces provided by
\cref{lem:trichromatic}, computes $\VD^*(F)$ in $\Otild(|F|)$ time. Thus, we establish the main theorem of this paper (note that we use $F$ to denote both the set of sites and the face to which they belong):

\begin{theorem}\label{thm:main}
    Given a planar graph $\X$, and a face $F$, one can preprocess $\X$ in $\Otild(|\X|)$ time such that given  additive weights $\omega: F \rightarrow \R^+$, one can compute $\VD^*(F)$ in $\Otild(|F|)$ time.
\end{theorem}

\noindent
{\bf Representing the Diagram.}
We use the definitions of the dual Voronoi diagram $\VD^*(F)$ from~\cite{ourJACM}.
As was done there, we assume that the graph we work with is triangulated, except for the single face $F$, whose vertices are exactly the set of sites of the diagram. We also assume that each site induces a non empty Voronoi cell (including at least the site itself).
Thus, $\VD^*(F)$ is a degree-3 tree with $O(|F|)$ nodes whose leaves are the copies of the dual vertex $F^*$ (corresponding to the face $F$).
Following~\cite{DBLP:journals/siamcomp/GawrychowskiKMS21}, we use the Doubly-Connected Edge List (DCEL) data structure for representing planar maps to represent the tree $\VD^*(F)$.

\medskip
\noindent
{\bf The Divide-and-Conquer Mechanism.}
Denote $|F| = k$.
We describe a divide-and-conquer algorithm for constructing $\VD^*(F)$ in $\Otild(k)$ time.
 If $k < 3$ then $\VD^*(F)$ contains no trichromatic vertices and its representation is trivial.
If $k = 3$, then $\VD^*(F)$ consists of either no trichromatic faces or just the single trichromatic face $\hat f$ obtained using \cref{sec:partition_to_trichromatic}.
When $k>3$ we partition the set $S$ of sites into two contiguous subsets along the face $F$ of (roughly) $k/2$ sites each.
For simplicity, we assume that each subset has size exactly $k/2$.
We call the sites $G = g_1, \dots, g_{k/2}$ in one subset the green sites, listed in counterclockwise order along $F$.
Similarly, we call the other subset $R = r_1, \dots, r_{k/2}$ the red sites, listed in clockwise order along $F$.
Note that the ordering is such that $g_1$ and $r_1$ are neighboring sites on the face  $F$.
We recursively compute $\VD^*(G)$ and $\VD^*(R)$, the Voronoi diagram of $G$ and of $R$, respectively.
We now describe how to merge these two diagrams into $\VD^*(F)$ in $\Otild(k)$ time.
The idea is similar to the stitching algorithm used in~\cite{DBLP:journals/siamcomp/GawrychowskiKMS21}. The main differences are that unlike~\cite{DBLP:journals/siamcomp/GawrychowskiKMS21} we have not precomputed the bisectors, but on the other hand, we utilize the point location mechanism of $\VD^*(G)$ and $\VD^*(R)$, which was not done in~\cite{DBLP:journals/siamcomp/GawrychowskiKMS21}.
In a nutshell, consider a super green vertex $G$ connected to all green sites with edges whose lengths correspond to the additive weight to each green site. Similarly, consider a super red vertex $R$.
Now consider the bisector $\beta^*(G,R)$.
$\VD^*(F)$ is obtained by cutting both $\VD^*(G)$ and $\VD^*(R)$ along $\beta^*(G,R)$, "glueing" the green side (according to $\beta^*(G,R)$) of $\VD^*(G)$ with the red side of $\VD^*(R)$.

\begin{figure}[htb]
     \centering
     \includegraphics[width=0.4\textwidth]{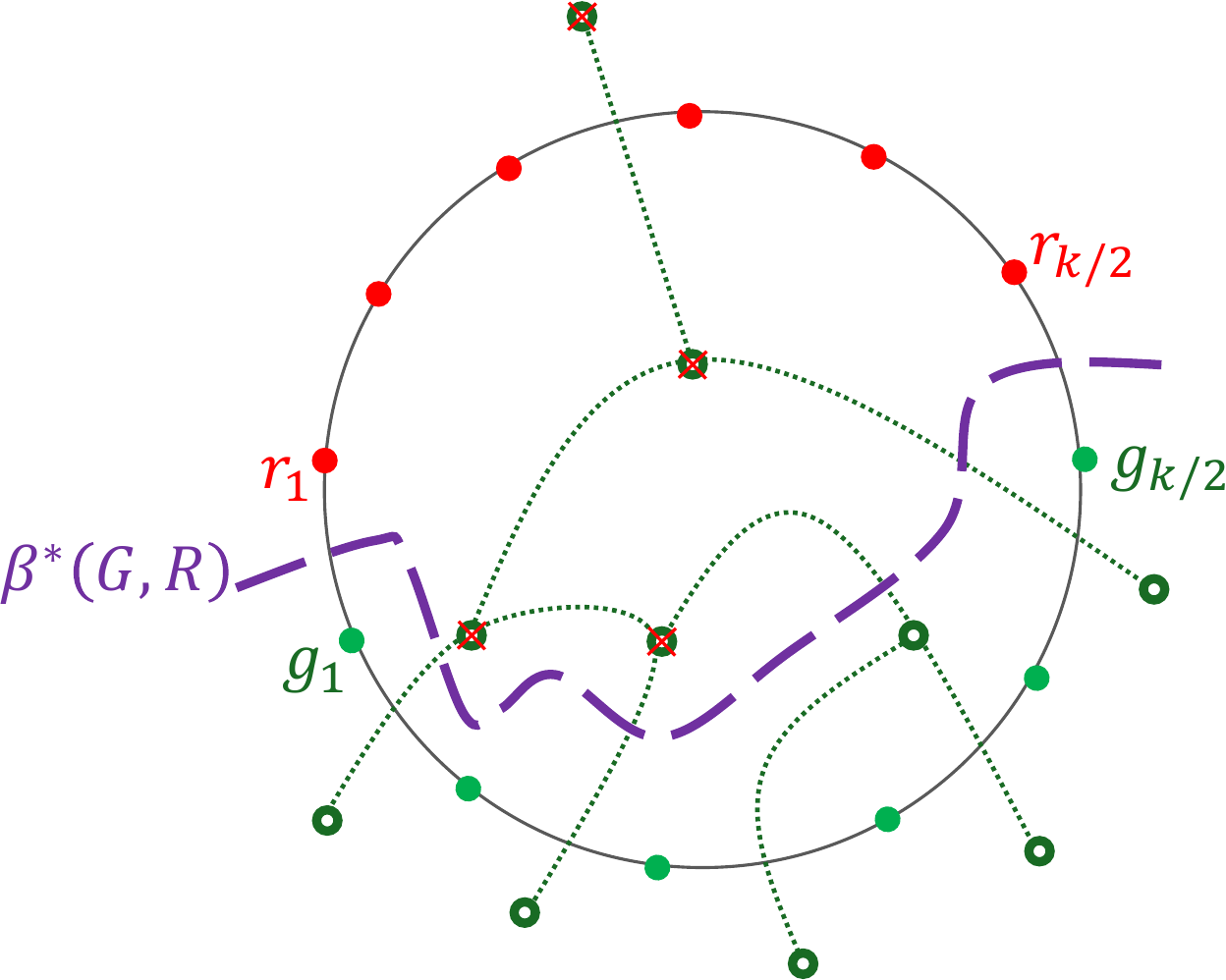}
     \hspace{0.5in}\includegraphics[width=0.4\textwidth]{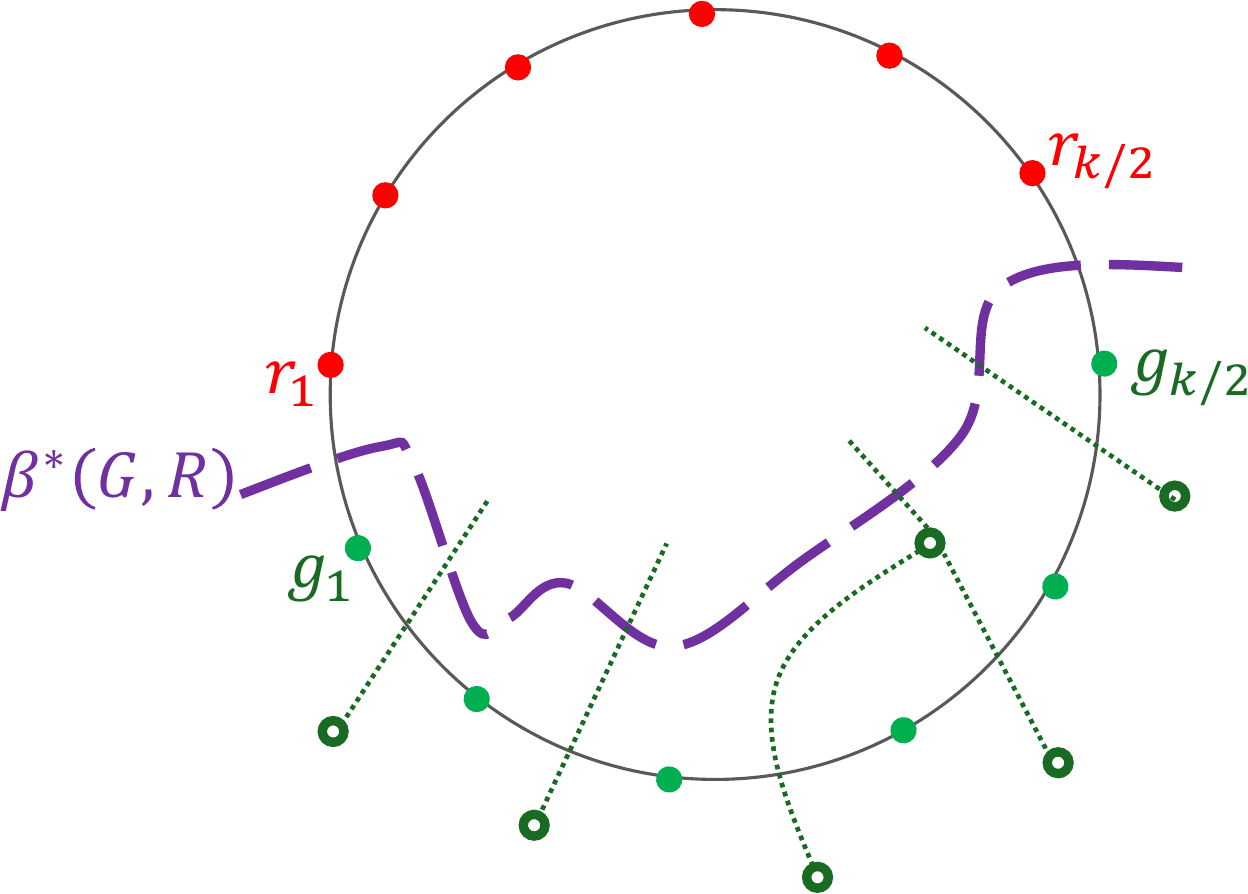}
     \caption{A schematic demonstration of the Voronoi vertex deletion process caused by the interaction of the $\beta^*(G,R)$ bisector with the green diagram $\VD^*(G)$. Left: the $\VD^*(G)$ diagram of the green sites, represented in dotted green lines, and the bisector $\beta^*(G,R)$ represented in a dashed purple line. The Voronoi (dual) vertices are represented in hollow green circles. The Vertices deleted in the process are crossed in red. Right: the remaining connected components and their respective dangling Voronoi edges after the deletion process.}
     \label{fig:Voronoid_Diagram_Calculation_Drawing_Deletion_Process}
 \end{figure}

This intuitive process can be performed efficiently using the following procedure.
To cut $\VD^*(G)$ along $\beta^*(G,R)$, we go over the $O(k/2)$ nodes of $\VD^*(G)$. Each of these nodes corresponds to a trichromatic face $f$ (or to an edge of the hole $F$ if the node is a leaf of $\VD^*(G)$).
For each vertex $v$ of $f$ (there are 3 or 2 such vertices depending on whether we handle an internal node or a leaf of $\VD^*(G)$), let $g_i$ be the green site closest (in the sense of additive distance) to $v$ (the identity of $g_i$ is stored explicitly in the representation of $\VD^*(G)$).
We perform a point location query for $v$ in $\VD^*(R)$, obtaining the red site $r_j$ closest (in the sense of additive distance) to $v$.
We compare the distance from $r_j$ to $v$ with the distance from $g_i$ to $v$ (these distances are available through the MSSP data structure for $F$). If the distance from $r_j$ is smaller, we know that $f$ is not a trichromatic face in $\VD^*(F)$, so we delete the corresponding node from $\VD^*(G)$.\footnote{To be precise, $f$ might be a trichromatic face of $\VD^*(F)$, but it is not an all-green trichromatic face in $\VD^*(F)$, so it is not contributed to $\VD^*(F)$ by $\VD^*(G)$.}

Some Voronoi edges of $\VD^*(G)$ have both their endpoints deleted by this process. These edges are entirely on the red side of $\beta^*(G,R)$, and do not participate in $\VD^*(F)$. Other edges have both their endpoints not deleted. These edges are entirely on the green side of $\beta^*(G,R)$, and are part of $\VD^*(F)$.
We call the Voronoi edges with one endpoint deleted and the other not deleted {\em dangling} edges.  The bisector $\beta^*(G,R)$ intersects $\VD^*(G)$ at the dangling edges.
We associate each dangling edge with the two green sites whose Voronoi cells (in $\VD^*(G)$) are on either side of the dangling edge.
The following lemma proves that there is at most a single dangling edge in each of the components of $\VD^*(G)$ obtained by the above deletion process.

\begin{lemma}\label{lem:merge_bisector_dangling_edge_components}
When the deletion process described above terminates, each surviving connected component of $\VD^*(G)$ contains at most a single dangling edge.

\end{lemma}

\begin{proof}
Assume towards contradiction that some connected component of $\VD^*(G)$ contains two edges $e_1, e_2$ crossed by $\beta^*(G,R)$.
Note that none of the endpoints of $e_1$ and $e_2$ that were not deleted by the process is a leaf (or else the connected component would consist of just that leaf).
Let $f_1$ (resp., $f_2$) be the dual vertex (primal face) in the intersection of the bisector corresponding to $e_1$ (resp., $e_2$) and $\beta^*(G,R)$.
Consider the cycle $C$ formed by the unique path in $\VD^*(G)$ between $f_1$ and $f_2$, and the portion of $\beta^*(G,R)$ between $f_1$ and $f_2$.
See \cref{fig:Voronoid_Diagram_Calculation_Drawing_Proof_First_Lemma} (left) for an illustration.
Observe that the cycle $C$ encloses no green sites because the path in $\VD^*(G)$ between $f_1$ and $f_2$ contains no leaves of $\VD^*(G)$, so it is disjoint from $F$, and because  the bisector $\beta^*(G,R)$ is disjoint from $F$ except for its first and last edge by our assumption that the sites are contiguous along $F$ and that every site is in its own Voronoi cell.
A contradiction now arises because the cycle $C$ encloses some primal vertex $v$ that belongs to a Voronoi cell of some green site $g_i$, but the shortest path from $g_i$ to $v$ cannot cross into $C$; It cannot cross any bisector of $\VD^*(G)$  because one endpoint of such a crossing edge does not belong to the cell of $g_i$ in $\VD^*(G)$.
It cannot cross $\beta^*(G,R)$  because one endpoint of such a crossing edge does not belong to the green cell of $\VD^*(\{G,R\})$.
\end{proof}

\begin{figure}[htb]
     \centering
    \includegraphics[width=0.4\textwidth, trim=0 -3cm 0 0]{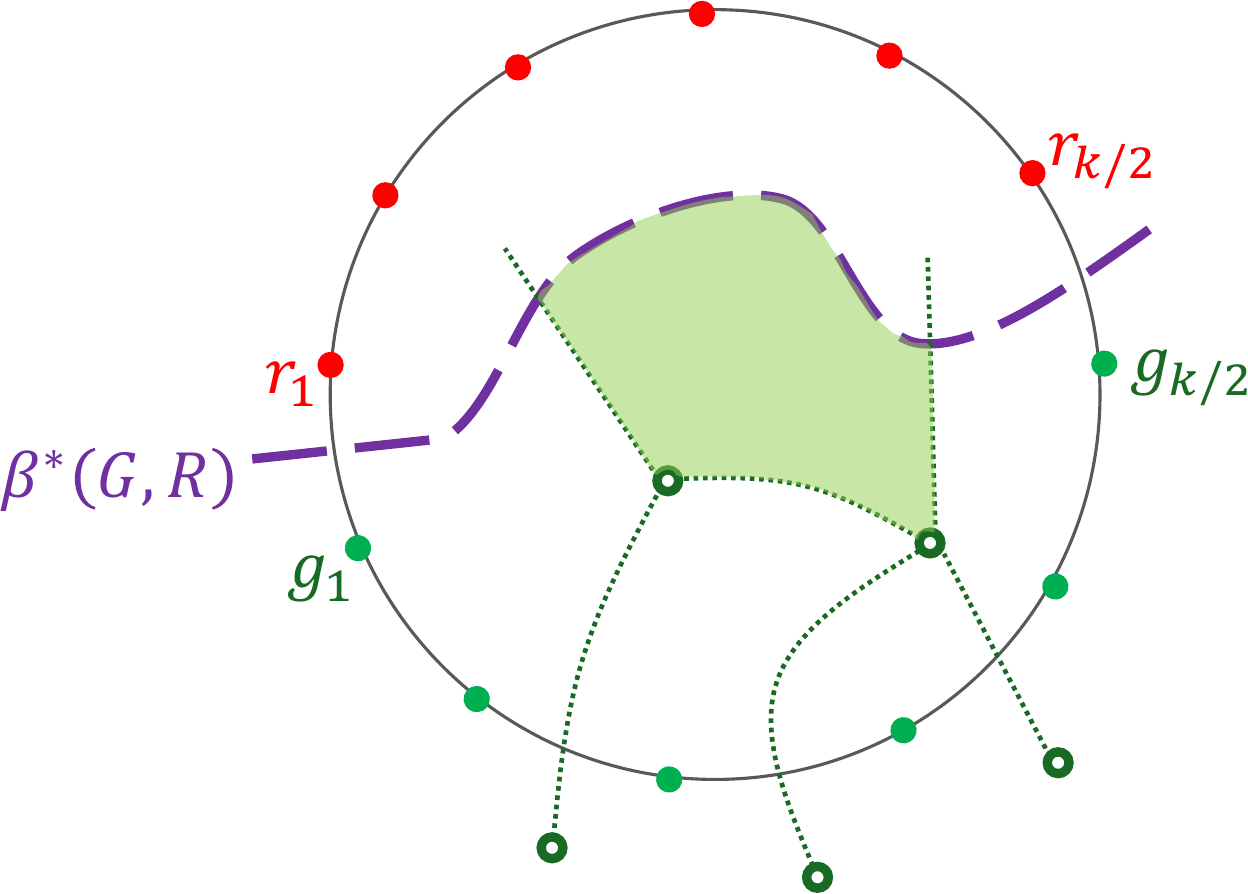}
     \hspace{0.4in}\includegraphics[width=0.45\textwidth]{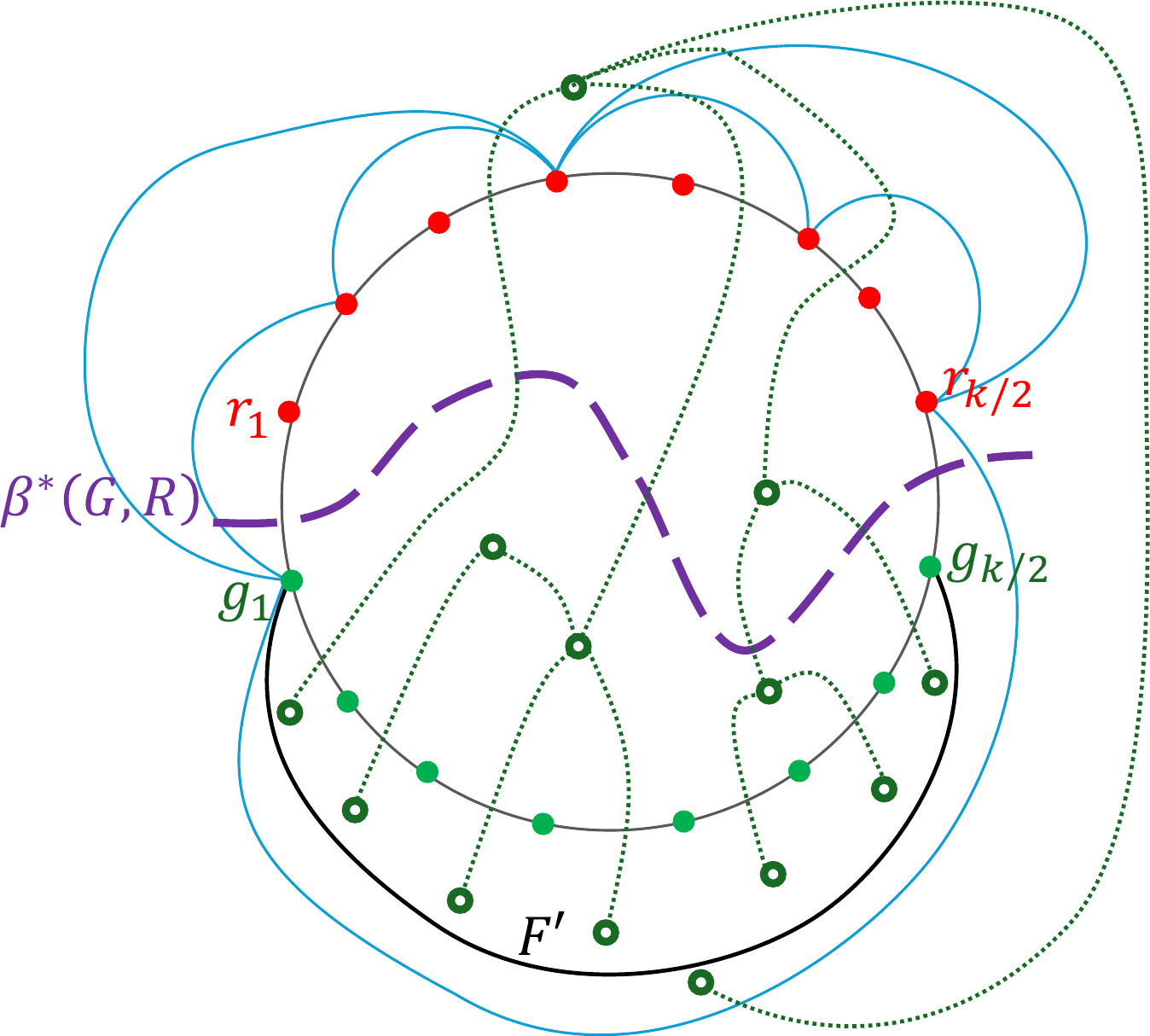}
     \caption{Left: The cycle $C$ (enclosing the shaded green area) for the contradiction in the proof of Lemma~\ref{lem:merge_bisector_dangling_edge_components}. Right: Enforcing the assumptions before the recursive call that computes $\VD^*(G)$. The bold black edge guarantees that the green sites are exactly the sites of a face (the face $F'$). The blue artificial edges guarantee triangulation.}
     \label{fig:Voronoid_Diagram_Calculation_Drawing_Proof_First_Lemma}
 \end{figure}

Having found the components of $\VD^*(G)$ (and by a similar process of $\VD^*(R)$) that form $\VD^*(F)$, we trace $\beta^*(G,R)$, and stitch the components of $\VD^*(G)$ and $\VD^*(R)$ at new trichromatic vertices that we identify along the way.
The first endpoint of $\beta^*(G,R)$ (which is a leaf of $\VD^*(F)$) is a copy of $F^*$ that lies at the end of a dual of the edge of $F$ separating between the sites $g_i=g_1$ and $r_j = r_1$.
The next trichromatic vertex along $\beta^*(G,R)$ occurs when $\beta^*(G,R)$ intersects  the boundary of the Voronoi cell of either $g_i$ or of $r_j$.

Recall that each dangling edge with the two green sites whose Voronoi cells (in $\VD^*(G)$) are on either side of the dangling edge. We shall prove in Lemma~\ref{lem:unique_connected_components},
that the sites $g_1$ and $g_{k/2}$ are associated with exactly one dangling edge, and all other sites are associated with either exactly two or exactly zero dangling edges.
To identify the next trichromatic vertex of $\VD^*(F)$ along $\beta^*(G,R)$, we inspect the dangling edge associated with $g_i$. When $g_i=g_1$, there is only one associated dangling edge. In general, there are two associated dangling edges, but only one was not yet handled by the stitching process.
 The dangling edge that was not yet handled represents a bisector between two green sites. One of them must be $g_i$ (since the intersection is a trichromatic face on the boundary of the Voronoi cell of $g_i$).
Denote the other one by $g_{i'}$. Similarly, a possible candidate for the next trichromatic vertex of $\VD^*(F)$ along $\beta^*(G,R)$ may come from the yet unhandled dangling edge of a connected component of $\VD^*(R)$ that is associated with $r_j$, which represents a bisector between $r_j$ and some other red site $r_{j'}$.

We use \cref{lem:trichromatic} to  find the two possible candidates: the trichromatic face $f_g$ of $(r_j, g_i, g_{i'})$ and the trichromatic face $f_r$ of $(r_j, r_{j'}, g_i)$. Only one of them is a true trichromatic vertex of $\VD^*(F)$, which can be decided by using the same point location process we had used in the deletion process above for each of its incident primal vertices.
Suppose w.l.o.g. that we identified that the next trichromatic face is the face $f_g$ of $(r_j, g_i, g_{i'})$ (the procedure for $f_r$ is symmetric).
We connect in $\VD^*(S)$ the previous trichromatic face on $\beta^*(G,R)$ with $f_g$ via a new Voronoi edge, and make $f_g$ the new endpoint of the dangling edge associated with $g_i$.

Next, we infer the identity of the edge of the $\beta^*(G,R)$ bisector leaving $f_g$ via which the traversal of $\beta^*(G,R)$ continues. Since  $f_g$ was the new trichromatic face, then $\beta^*(G,R)$ just crossed in $\VD^*(G)$ from the cell of $g_i$ to the cell of $g_{i'}$, so we repeat the process of finding the next trichromatic face along $\beta^*(G,R)$ with sites $r_j$ and $g_i'$.  The process terminates when it has handled all of the dangling edges in all the components of $\VD^*(G)$ and $\VD^*(R)$.

To complete the correctness argument it remains to prove the bound on the association between sites and dangling edges (\cref{lem:unique_connected_components}). However, before doing that, we need to elaborate on a technical issue we had glossed over in the description of the recursive approach.
We had assumed that in any $\VD^*$ in the recursion, the sites are exactly the vertices of some face $F$ of the graph $\X$  in which $\VD^*$ is computed, and that $F$ is the only face of $\X$ that is not a triangle.
To satisfy this requirement, before making the recursive call that computes $\VD^*(G)$, we add to $\X$ an artificial infinite-length edge connecting $g_1$ and $g_{k/2}$. This artificial edge is embedded in the face $F$, splitting it into two new faces $F'$ and $F''$, such that the vertices of $F'$ are exactly the green sites.
We triangulate $F''$ with infinite length edges.
Let $\X'$ denote the resulting graph.
See \cref{fig:Voronoid_Diagram_Calculation_Drawing_Proof_First_Lemma} (right).
Note that $\X'$ now satisfies the assumptions so we can invoke the construction algorithm recursively on $\X'$ and obtain $\VD^*(G)$.
In $\VD^*(G)$, the boundary of the Voronoi cell of each $g_i$ forms a simple path between the two leaves (copies of $F'$ corresponding to dual edges of $F'$ between $g_{i-1}g_i$ and $g_i g_{i+1}$) in $\VD^*(G)$, which we associate with $g_i$.
Note that for sites $g_i$ for $1<i<k/2$, both associated leaves are real edges of $P$, whereas for $g_1$ and $g_{k/2}$ one leaf is real, but the other is dual to an artificial edge of $\X'$.

Since $\VD^*(G)$ may contain (dual) artificial vertices (i.e., faces $\X'$ that are not faces of $\X$), at the beginning of the deletion process we delete all artificial vertices of $\VD^*(G)$, and then apply the deletion process described above.
Observe that exactly one of the leaves corresponding to $g_1$ and $g_{k/2}$ is deleted by the deletion process, and none of the two leaves associated with the other sites is deleted. This is because each $G_i$ is in its own Voronoi cell in $\VD^*(F)$.

\begin{lemma}\label{lem:unique_connected_components}
When the deletion process of $\VD^*(G)$ described above terminates, sites $g_1$ and $g_{k/2}$ are associated with exactly one dangling edge, and all other sites are associated with either two or zero dangling edges.
\end{lemma}

\begin{proof}
The deletion process deletes a single contiguous subpath from the boundary of each Voronoi cell. This is because otherwise, there would be a resulting connected component that contains two dangling edges, contradicting \cref{lem:merge_bisector_dangling_edge_components}.
Thus, since the dangling edges associated with site $g_i$ are the dangling edges on the boundary of the Voronoi cell of $g_i$, each site is associated with at most two dangling edges.
Since for $g_1$ and $g_{k/2}$ exactly one of the corresponding leaves are deleted, the process results in a single dangling edge associated with each of these two sites.
For all the other sites, none of the two corresponding leaves are deleted, hence there are two associated dangling edges with  each site other than $g_1$ and $G|_{k/2}$.
\end{proof}

\bibliography{bib}
\bibliographystyle{abbrv}

\appendix
\section{The Enhanced MSSP Data Structure (proof of \cref{lem:enhanced_mssp})}\label{sec:enhanced_mssp}

In this section we prove \cref{lem:enhanced_mssp} by extending the seminal MSSP data structure\cite{Klein02,CabelloCE13}.
Let $s_1,\ldots,s_{|f|}$ be the vertices of the face $f$ in cyclic order.
The MSSP is initialized with the shortest paths tree $T_1$ rooted at $s_1$ (computed in $O(n\log n)$ time using Dijkstra). Then, $T_2$ is computed from $T_1$ using {\em pivots}. A pivot is the process of replacing an edge of $T_1$ with an edge not in $T_1$.
Then, $T_3$ is computed from $T_2$ and so on until $T_{|f|}$ is computed.
Klein showed that: (1) over the entire process of computing $T_1,\ldots,T_{|f|}$ there are only $O(n)$ pivots, (2) every pivot can be found and executed in $O(\log n)$ time by maintaining the dual tree of the current $T_k$ in a dynamic tree data structure \cite{SleatorT83}.
Finally, by using the persistence technique of \cite{DriscollSST89}, all the $T_k$'s can be recorded in $(n \log n)$ space so as to permit $\dist(s,v)$ queries.
We slightly modify the MSSP data structure, such that the dynamic trees would be implemented with top-trees of Alstrup et al.~\cite{AHLT05}.
This retains all features of the original implementation while additionally supporting $\ancestor(s,v,d)$ queries in  $O(\log n)$ time per query.

In order to implement $\direction(s,v,P)$ queries, we additionally maintain a persistent~\cite{DriscollSST89} range data structure~\cite{Willard85,Chazelle88} $Left_P$  for every shortest path $P$ in the decomposition tree.
The data structure $Left_P$ stores, throughout the running of MSSP, the set of vertices (ordered by their order on $P$) that the current source vertex $s_i$ reaches from the left.
    Once $Left_P$ is maintained properly, $\countA(s,P,i,j)$ queries can be supported in $\Otild(1)$ time by accessing $Left_P$ at the time in which the stored MSSP tree is the shortest path tree rooted at $s$.
    Moreover, $\select$ queries can be supported via binary search that utilizes $\countA$ queries.

    Initially, $Left_P$ is obtained in $O(|P|)$ time by checking the parent of every vertex $v\in P$ in the shortest paths tree of $s_1$ (if the parent is on $P$, then $v$ inherits the direction from its parent).
    In order to maintain $Left_P$, we store an additional range data structure $Cross(u)$ for every vertex $u\in G$.
    For a vertex $u$, let $e_1,e_2, \ldots e_d$ be the edges incident to $u$ in clockwise order.
    For every separator path  $P$ such that $u\in P$, there is a cyclic interval of the edges of $u$ that are on the left of $P$, and a cyclic interval that is on the right.
    The data structure $Cross(u)$ stores the left and right cyclic intervals of all paths $P$ containing $u$ such that given two edges $e_i$ and $e_j$ incoming to $u$, it returns the set of paths $P$ such that $e_i$ and $e_j$ are on different sides of $P$.
    Since each vertex $u$ of $G$ has constant degree, implementing  the $Cross(u)$ data structures is trivial and requires $\Otild(|P|)$ time for each $P$.
    This sums up to $\Otild(n)$.

To maintain $Left_P$ during the MSSP execution, while obtaining $T_{k}$ from $T_{k-1}$, whenever a pivot is performed, we apply the following:
    Let $u$ be the vertex such that the pivot changes $u$'s incoming edge from $e_i$ to $e_j$.
    It is possible that the Left/Right status of $u$ has been changed with respect to (possibly several) paths $P$ containing $u$.
    Notice that if such a status change occurs for $u$, it also occurs for all descendants of $u$ in $T_{k}$ that lie on $P$ (since $P$ is a shortest path, and due to the uniqueness of shortest paths, those descendants all lie on a single subpath of $P$ that contains $u$).
   We need to find the set of paths $P$ for which the Left/Right status of $u$ changes.
    We do this by querying $Cross(u)$ to obtain all paths $P$ such that $e_i$ and $e_j$ are on different sides of $P$.
    If $e_i$ or $e_j$ are on $P$, the new/old type of $u$ is determined by its first ancestor in $T_{k}$/$T_{k-1}$ that is not on $P$ (which can be found by binary search on $P$ and the tree edges of $T_{k}$/$T_{k-1}$).
    If the type of $u$ changes, we remove/insert $u$ from $Left_P$ (depending on the new type).
   We also run the following procedure in order to update the Left/Right status of all descendants of $u$ in $P$, if necessary.
    Starting from $u$, we traverse $P$ twice, once in each direction of $P$.
    As long as a descendant of $u$ is met, we update the traversed vertex with the new type of $u$.
    We halt when we reach a vertex of $P$ that is not an ancestor of $u$.

    \medskip
\noindent
{\bf Correctness:}
    We claim that all $Left_P$ data structures are updated correctly after every pivot.
    Clearly, if both $e_i$ and $e_j$ are left or right edges, the left/right status of the vertex $u$ does not change, so only the set of paths returned from $Cross(u)$ needs to be checked.
    This set is updated naively - both $u$ and all of its decedents on $P$ are updated one by one.
    Due to the uniqueness of shortest paths, the decedents of $u$ on $P$ form a consecutive subpath of $P$.
    Therefore, the algorithm did not miss any descendants of $u$ on $P$ as a result of halting when encountering a non-descendant in each direction.
    Finally, vertices that are not decedents of $u$ did not have their shortest path change as a result of the pivot, and therefore their Left/Right status did not change with respect to any path $P$.

    \medskip
\noindent
{\bf Running time.}
    The algorithm queries a $Cross(u)$ data structure per pivot, which sums up to $\Otild(n)$ across the run of MSSP.
    As for the updates to $Left_P$, we claim that every vertex $u$ is accessed in this way $O(1)$ times per path $P$ containing it.
    Let us consider the event in which a vertex $u$ is accessed because $e_j$ is not of the same type as $e_i$ with respect to the path $P$.
    We distinguish between two types of vertex access events: an event in which $u$ changes its Left/Right status with respect $P$, and an event in which it does not.
    If $u$ does not change its status, we necessarily have that either the old edge $e_i$ or the new edge $e_j$ is on the path $P$ (as otherwise one is a right edge and the other is a left edge, leading to a status change of $u$ with respect to $P$).
    It immediately follows from the fact that every edge participates in $O(1)$ pivots that this event occurs at most once per $u\in P$, and therefore the number of occurrences of this event is proportional to the total length of all paths $P$, which is $\Otild(n)$.
   If on the other hand $u$ does change its status, then by the following claim this can happen at most two times throughout the running of the algorithm.

    \begin{claim}\label{clm:cyclicleftright}
        Let $P$ be a separator path and let $u$ be a vertex of $P$.
        The set of source vertices $s_k$ such that $s_k$ reaches $u$ from the left with respect to $P$ form a cyclic interval on the face $f$.
    \end{claim}
    \begin{proof}
        Let $s_i$ and $s_j$ be two sites such that $s_i$ and $s_j$ reaches $u$ from the left.
        We prove that either all vertices in the cyclic interval $f[s_i,s_j]$ reach $s$ from the left, or all vertices in the cyclic interval $f[s_j,s_i]$ reach $s$ from the left, which leads to the claim.
        Consider the cycle $C$ formed by (1) the shortest path $R_i$ from $s_i$ to $u$, (2) the shortest path $R_j$ from $s_j$ to $u$,  and (3) the Jordan Curve connecting $s_i$ and $s_j$ embedded on $f$.
        Clearly, all right edges of $P$ are on one side of the cycle $C$, as $R_i$ and $R_j$ enter $P$ from the left and do not cross $P$ due to uniqueness of shortest paths.
        Let us denote the side of the cycle that contains the right edges of $P$ as the `out' side, and the other as the `in' side.
        Consider the case in which $f(s_i,s_j)$ is on the `in' side of $C$ (the other case is symmetric).
        We claim that every vertex in $s_k \in f(s_i,s_j)$  reaches $u$ from the left.
        Assume to the contrary that there is a vertex $s_k \in f(s_i,s_j)$ that reaches $u$ from the right.
        Since the shortest path $R_k$ from $s_k$ to $u$ starts on the `in' side of $C$, and enters $P$ from the `out' side of $C$, it must cross the $C$ before reaching a vertex of $P$.
        In particular, it must either cross $R_i$ or $R_j$ (say, $R_i$) at some vertex $z$.
        Then, according to the uniqueness of shortest paths, we have that $R_k[z,u]=R_i[z,u]$, which means that $s_k$ reaches $u$ from the left, a contradiction.
    \end{proof}

    To summarize, it follows from \cref{clm:cyclicleftright} that $u$ changes its Left/Right status with respect to $P$ at most twice throughout the algorithm (namely, when the boundaries of the cyclic interval indicated by \cref{clm:cyclicleftright} are processed).
    Therefore, the number of times the algorithm processes a vertex that changes its status is bounded by the total size of all paths in the decomposition, which is $O(n)$.
    Since a vertex is processed in $\Otild
    (1)$ time, the total running time is bounded by $\Otild(n)$ as required.

\section{Proof of \cref{lem:Hleft-is-enough}}\label{app:missing_proof}

    Assume to the contrary that an $(H_{\leftside},\leftside)$-partition is not an $(H,\leftside)$-partition.
    In particular, it must be that there exists some $v\in P$ such that $v$ $(H_{\leftside},\leftside)$-likes some site $s_{\leftside}\in H_{\leftside}$ but does not $(H,\leftside)$-like $s_{\leftside}$.
    This means that for some $c\in H$ we have  that (1) $v\in\Vor_{H\cup\{s_x,s_y\}}(c)$,  (2) $c\ne s_{\leftside}$, and (3)  $v\in\Left(P,c)$.

    We first claim that $c\in H\setminus H_{\leftside}$.
    Assume to the contrary that $c\in H_{\leftside}$, and recall that $v\in\Left(P,c)$. Therefore, it must be that $v\in P_c$ in the $(H_{\leftside},\leftside)$-partition.
    Since $P_c\cap P_{s_{\leftside}}=\emptyset$, we obtain a contradiction.
    Thus, $c\in H\setminus H_{\leftside}$.

    Consider the cycle $C$ obtained by concatenating $R_y, P, R_x$  and the Jordan curve connecting $s_x$ and $s_y$ embedded in the face $F$.
    We think of $C$ as an oriented cycle whose orientation is consistent with that of $P$ so as to define left and right properly.
    Notice that $R_x$ and $R_y$ do not cross each other. Therefore, $C$ is non-self crossing.
    Since $c\in H\setminus H_{\leftside}$, $c$ is to the right of $C$. Therefore, $R_{c,v}$ must cross $C$ (recall that $v\in\Left(P,c)$).
    By uniqueness of shortest paths $R_{c,v}$ does not cross $P$. Therefore, it must cross either $R_x$ or $R_y$.
    Contradicting $v\in \Vor_H(c)$.

\section{Partition of Two Sites with Swirly Paths}\label{app:two_sites}
In this section we prove \cref{lem:s1s2partition}.
We first (in \cref{sec:xy-partition}) prove the lemma for the special case where $s_1=s_x$ and $s_2=s_y$.
Then (in \cref{sec:two-partition}), we utilize the partition obtained for $s_x,s_y$ to compute a partition for any pair $s_1,s_2\in F_{\leftside}$.

\subsection{Partition for $s_x$ and $s_y$}\label{sec:xy-partition}

To prove \cref{lem:s1s2partition} we begin with the simpler case of  $s_1=s_x$ and $s_2=s_y$.
We start by proving the following structure regarding $s_x$ and $s_y$.

The following lemma shows that for $s\in\{s_x,s_y\}$ we can recognize swirly and non-swirly paths to vertices in $\Left(P,s)$ by comparing their position in the shortest paths tree of $s$ compared to $x$ or $y$.
\begin{lemma}\label{lem:rightloseleftswirl}
    Let $v \in \Left(P,s_y)$ be a vertex,
    if $v$ is to the right of $y$ in the shortest paths tree of $s_y$, then $\dist(s_y,v) > \dist(s_x,v)$ and $v\in \Left(P,s_x)$.
    Otherwise, $s_y$ reaches $v$ via a non-swirly path.

    Similarly, for a vertex $v\in\Left(P,s_x)$, if $v$ is to the left of $x$ in the shortest paths tree of $s_x$, then $\dist(s_x,v) > \dist(s_y,v)$ and $v\in \Left(P,s_y)$.
    Otherwise, $s_x$ reaches $v$ via a non-swirly path.
\end{lemma}
\begin{proof}

\begin{figure}[ht]
    \centering
    \begin{subfigure}[t]{0.3\textwidth}
        \centering
        \includegraphics[width=\textwidth]{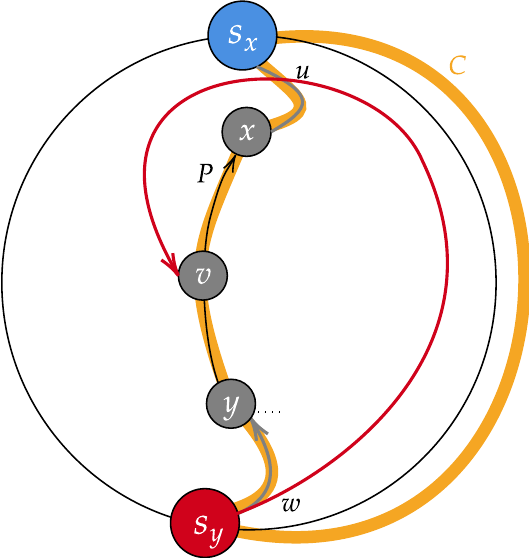}
        \caption{$C$}
    \end{subfigure}
    \hspace{1in}
    \begin{subfigure}[t]{0.3\textwidth}
        \centering
        \includegraphics[width=\textwidth]{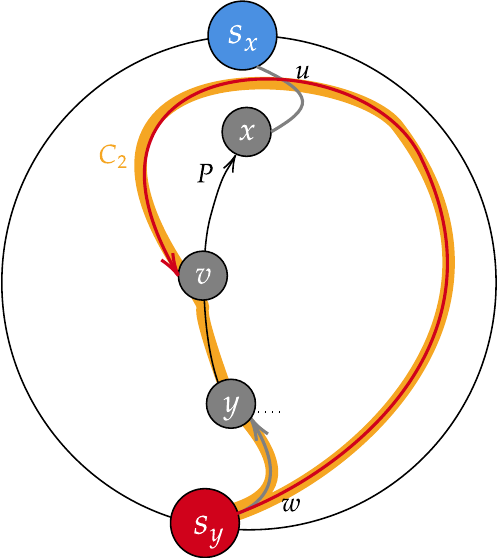}
        \caption{$C_2$}
    \end{subfigure}
    \caption{The cycles for the proof of \cref{lem:rightloseleftswirl}.}
    \label{fig:XY-left-swirl}
\end{figure}
    We only prove the first statement, the second statement is symmetric.
    Consider the cycle $C$ obtained by concatenating $R_y$, $P$,$R_x$ and the Jordan curve connecting $s_x$ and $s_y$ embedded in the face $F$.
    Notice that $R_x\cap R_y=\emptyset$ (since we have $\dist(s_x,z)<\dist(s_y,z)$ and $\dist(s_x,t)>\dist(s_y,t)$ for every $z\in R_x$, $t\in R_y$) and therefore $C$ is a non-self-crossing cycle.
    We think of $C$ as an oriented consistently with $P$ so as to define left and right properly.
    Notice that all edges entering $P\setminus\{x,y\}$ from the left are on the left side of $C$.
    Let $(w',w)$ be the first edge on $R_{s_y,v}$ that is not on $R_y$.
    We start by proving the first claim, where $v$ is to the right of $y$ in the shortest paths tree of $s_y$.
    We distinguish between two cases.
    If $w'\in P$, then $R_{s_y,v}[w',v]=P[w',v]$ and in this case $R_y$ reaches $P$ from the right and therefore $R_{s_y,v}$ reaches $V$ from the right, a contradiction.
    If $w'\notin P$, $w$ is on the right side of $C$.
    Let $Q=R_{s_2,v}[w,v]$.
    Since $v\in C$, $Q$ must intersect $C$.
    Let $u$ be the first vertex in $Q$ that is on $C$.
    Clearly, $u\notin R_y$ by uniqueness of shortest paths.
    Moreover, $u\notin P\setminus(R_x\cup R_y)$ since $R_{s_y,v}$ enters $P$ from the left.
    It follows that $u\in R_x$ and therefore $\dist(s_x,v)<\dist(s_y,v)$ as required.

    We next prove that $v\in \Left(P,s_x)$.
    Consider the cycle $C_2$ composed of the concatenation of $R_y$, $P[y,v]$ and $R_{v,s_y}$.
    We think of $C_2$ as an oriented cycle whose orientation is consistent with that of $P$ so as to define left and right properly.
    Notice that $C_2$ has no self crossing.
    Notice that $s_x$ is on the left side of $C_2$  and every right edge of $P$ is on the right side of $C_2$.
    Assume by contradiction that $v\notin\Left(P,s_x)$, then $R_{s_x,v}$ must cross the $C_2$.
    By uniqueness of shortest paths, $R_{s_x,v}$ does not cross $P$ nor $R_{s_y,v}$.
    Thus, $R_{s_x,v}$ must cross $R_y$.
    Let $z$ be a vertex on $R_y\cap R_{s_x,v}$.
    Since $z\in R_y$ we have $\dist(s_y,z)<\dist(s_x,z)$ and therefore $\dist(s_y,v)\le\dist(s_y,z)+\dist(z,v)<\dist(s_x,z)+\dist(z,v)=\dist(s_x,v)$ a contradiction to what we already proved.
    Therefore it must be that $v\in \Left(s_x,v)$.

    Regarding the second claim, if $v$ is not to the right of $y$ in the shortest path  tree of $s_y$ then it must be either to the left, an ancestor of $y$, or a decedent of $y$.
    Clearly, the case of $v$ being an ancestor of $y$, then $R_{s_y,v}$ is a subpath of $R_y$, that is disjoint from $R_x$ and therefore it is non-swirly.
    In addition, $v$ being a decedent of $y$ means that $R_{s_y,v}=R_y\circ P[y,v]$ which crosses neither $R_x$ nor $R_y$.
    Finally, consider the case where $v$ is to the left of $y$ in the shortest path tree of $s_y$.
    If $w'\in P$ then $R_{s_y,v}=R_y[s_y,w]\circ P[w,v]$ which crosses neither $R_x$ nor $R_y$.
    Otherwise, $w$ must be in the left side of $C$.
    By uniqueness of shortest paths, $Q=R[w,v]$ cannot cross $R_y$ or $P$.
    Assume to the contrary that $Q$ crosses $R_x$, and let $t$ be the first vertex on $Q$ after crossing $R_x$.
    Notice that $t$ is on the right side of $C$.
    The path $Q'=Q[t,v]$ cannot intersect with $R_x$ or $R_y$ by uniqueness of shortest paths.
    Therefore $Q'$ must enters $P$ from the right, a contradiction.
\end{proof}

\begin{lemma}\label{lem:sxsyswirlypartition}
    Given $s_x$ and $s_y$, there is an algorithm that outputs in $\Otild(1)$-time a partition of $P$ into two (possibly empty) parts $P=P^x_1\circ P^x_2$  such that:
    \begin{enumerate}
        \item Every vertex $v\in\Left(P^x_1,s_x)$ is reached from $s_x$ via a non-swirly path.
        \item For every vertex $v\in \Left(P^x_2,s_x)$ we have $\dist(s_x,v)>\dist(s_y,v)$.
    \end{enumerate}
    Similarly, the algorithm outputs a partition of $P$ into two (possibly empty) parts $P=P^y_1\circ P^y_2$  such that:
    \begin{enumerate}
        \item For every vertex $v\in \Left(P^y_1,s_y)$ we have  $\dist(s_y,v)>\dist(s_x,v)$.
        \item Every vertex $v\in \Left(P^y_2,s_y)$ is reached from $s_y$ via a non-swirly path.
    \end{enumerate}
\end{lemma}
\begin{figure}[ht]
    \centering
    \begin{subfigure}[t]{0.3\textwidth}
        \centering
        \includegraphics[width=\textwidth]{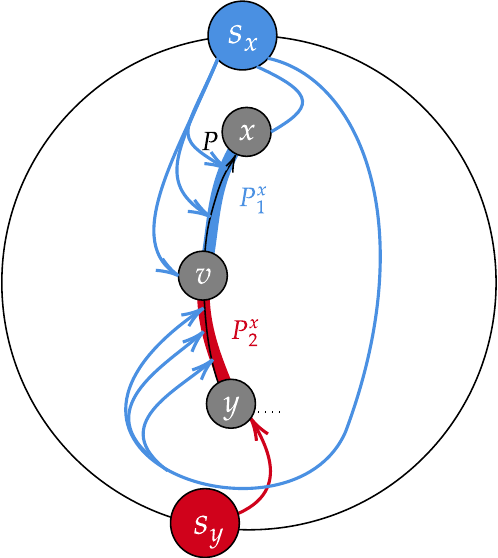}
        \caption{Partition of $P$ to $P^x_1$ and $P^x_2$.}
    \end{subfigure}
    \hspace{1in}
    \begin{subfigure}[t]{0.3\textwidth}
        \centering
        \includegraphics[width=\textwidth]{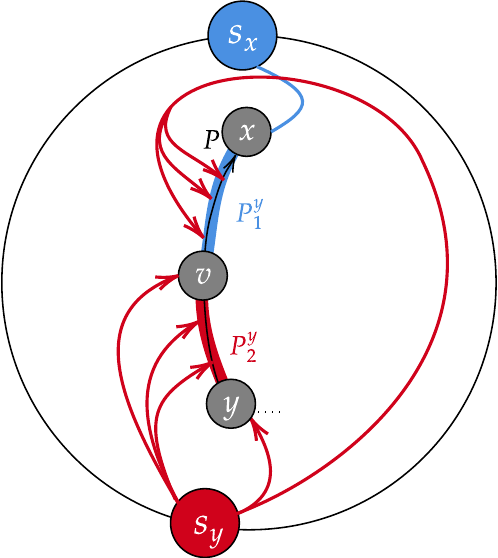}
        \caption{Partition of $P$ to $P^y_1$ and $P^y_2$.}
        \label{fig:Py12}
    \end{subfigure}
    \caption{Figures for the statement of  \cref{lem:sxsyswirlypartition}.}
    \label{fig:Px12Py12}
\end{figure}
\begin{proof}
    We prove the second claim of the lemma, the proof of the first claim is similar.

    The algorithm works in a binary search fashion to find a partition $P=P^y_1 \circ P^y_2$ such that every vertex in $\Left(P^y_1,s_y)$ is to the right of $y$ in the shortest paths tree of $s_y$, and every vertex in $\Left(P^y_2,s_y)$ is not to the right of $y$ in the shortest paths tree of $s_y$.
    We first claim that such a partition exists (\cref{clm:xypartitionexsist}), and then we show that this partition satisfies the conditions of the lemma.

    The following claim shows that the sought partition exists (see \cref{fig:Py12}).
\begin{claim}\label{clm:xypartitionexsist}
    There is a vertex $v\in P$ such that every $u\in \Left(P[y,v],s_y)$ is not to the right of $y$ in the shortest path tree of $s_y$, and every vertex in $\Left(P(v,x],s_y)$ is to the right of $y$ in the shortest paths tree of $s_y$.
\end{claim}
\begin{proof}
    Let $v$ be the last vertex on $\Left(P,s_y)$ that is not to the right of $y$ in the shortest paths tree of $s_y$.
    Notice that $v$ is well defined, as $y\in P$ is not to the right of $y$.
    By definition of $v$, every vertex $u\in \Left(P(v,x],s_y)$ is to the right of $y$ in the shortest paths tree of $s_y$.
    It remains to show that every vertex in $\Left(P[y,v],s_y)$ satisfies the claim.
    Assume to the contrary that there is a vertex $u\in \Left(P[y,v],s_y)$ that is to the right of $y$ in the shortest paths tree of $s_y$.
\begin{figure}[ht]
    \centering
    \includegraphics[width=0.3\textwidth]{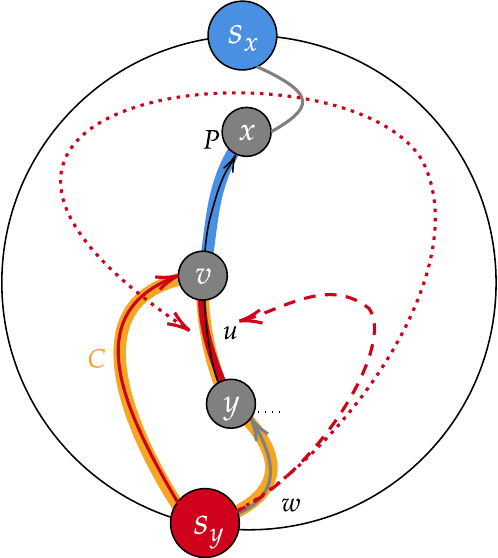}
    \caption{The (orange) cycle $C$ in the proof of \cref{clm:xypartitionexsist}. The dotted red path is impossible since it implies that $u$ is not to the right of $y$ in the shortest paths tree of $s_y$. The dashed red path is impossible since it implies that $u$ is reached from the right.}
   \label{fig:xypartitionexsist}
\end{figure}

    Consider the cycle $C$ that consists of the concatenation of $R_y$, $P[y,v]$, and $R_{v,s_y}$ in that order.
    We consider $C$ to be oriented in a way that is consistent with $R_y$ (from $s_y$ to $y$) and with $P$.
    Let $(w',w)$ be the first edge of $R_{s_y,u}$ that is not in $R_y$.
    If $w' \in P$, uniqueness of shortest paths implies that $R_{s_y,u}[w',u]=P[w,u]$ and $R_y[w',y]=P[w',y]$.
    Since $u$ is to the right of $y$, we get that $R_y[s_y,w'] = R_{s_y,u}[s_y,w']$ must enter $P$ from the right, a contradiction to $u$ being reached from the left by $s_y$.
    We therefore have that $w' \notin P$, and therefore $w$ is strictly on the right side of $C$.
    The path $R=R_{s_y,u}[w,u]$ must intersect the cycle $C$.
    Due to uniqueness of shortest paths, $R$ cannot intersect $R_y$.
    If $R$ intersects $R_{s_y,v} \setminus R_y$, uniqueness of shortest paths implies that $R_{s_y,u}$ agrees with $R_{s_y,v}$ at least on the first edge that is in $R_{s_y,v}$ and not on $R_y$.
    This leads to a contradiction, as it indicates that $u$ is not to the right of $y$ (in particular, $u$ has the same tree relationship with $y$ as $v$).
    The only remaining option is for $R$ to intersect $P \setminus (R_y \cup R_{s_y,v})$, but all the left edges of $P \setminus (R_y \cup R_{s_y,v})$ are on the left side of $C$, so this leads to $u$ being reached from the right, a contradiction.
\end{proof}

With \cref{lem:rightloseleftswirl,clm:xypartitionexsist}, we are finally ready to describe the algorithm.
The algorithm simply applies a binary search to find the vertex $v$ specified in \cref{clm:xypartitionexsist}.
Using $\countA$ and $\select$ queries of the enhanced MSSP, the algorithm finds the median vertex $u\in \Left(P,s_y)$, and checks if $u$ is to the right of $y$ in the shortest paths tree of $s_y$.
According to the answer, the algorithm decides if $u$ is before or after the sought vertex $v$ in $P$ and proceeds to the relevant half of $P$.
The recursion halts when there is only a single vertex in $\Left(P',s_y)$.
Clearly, the recursion has logarithmic depth, and each recursive call is implemented in $\Otild(1)$ time.
Upon finding $v$, the algorithm returns the partition $P^y_1 = P[y,v]$ and $P^y_2=P(v,x]$.
It follows directly from \cref{clm:xypartitionexsist,lem:rightloseleftswirl} that this partition satisfies the conditions of the lemma.
\end{proof}

\begin{lemma}\label{lem:out_partition_xy}
    There is an $\Otild(1)$ algorithm that given $s_x$ and $s_y$ outputs an $(\{ s_x,s_y \},\leftside)$-partition of $P$.
\end{lemma}
\begin{proof}
    The algorithm starts by applying \cref{lem:sxsyswirlypartition} to both $s_x$ and $s_y$, obtaining partitions $P= P^x_1 \circ P^x_2 = P^y_1 \circ P^y_2$.
    Let $a,b$ be the unique vertices such that $P^y_1=P[x,a)$ and $P^x_2 = P(b,y]$.
    Notice that every vertex in $P[a,b]=P^y_2 \cap P^x_1$ that is reached from the left by $s_x$ or by $s_y$, is reached from the left via a non-swirly path.
    Therefore, \cref{lem:out_partition_non-swirly} can be applied to $P[a,b]$.
    The algorithm applies \cref{lem:out_partition_non-swirly} to $P[a,b]$, obtaining a partition $P[a,b] = \hat P_1 \circ \hat P_2$ and returns $P_x = P^y_1 \circ \hat P_2$
    and $P_y = \hat P_2 \circ P^x_2$.
    Clearly, the algorithm takes $\Otild(1)$ time.
    The correctness follows from the correctness of \cref{lem:out_partition_non-swirly} and from the fact that every vertex $v\in P^y_1$ that $s_y$ reaches from the left has $\dist(s_x,v) < \dist(s_y,v)$ and therefore $(\{s_x,s_y\},\leftside)$-likes $s_x$ and every vertex $v\in P^x_2$ that $s_x$ reaches from the left has $\dist(s_y,v) < \dist(s_x,b)$ and therefore $(\{s_x,s_y\},\leftside)$-likes $s_y$, as required.
\end{proof}

\subsection{Partition for $s_1,s_2\in F_{\leftside}$}\label{sec:two-partition}

In this section we show how to compute the relaxed partition for any pair of vertices $s_1,s_2\in F_{\leftside}$ in $\Otild(1)$ time, thus proving \cref{lem:s1s2partition}.
The main idea is to recognize, separately for each $s\in \{s_1,s_2\}$, one of two cases: Either that $s$ is not the site of any vertex in $\Left(P,s)$ which means we can assign all $P$ to the other site.
Or, recognize a partition of $\Left(P,s)$ into up to two intervals, where one interval contains only vertices $v$ with $R_{s,v}$ being non-swirly, and the other interval with all vertices $v$ with $R_{s,v}$  being swirly.
To use these partition we exploit the partition of $P$ with respect to $\{s_x,s_y\}$ (\cref{sec:xy-partition}).
We then show that when considering both $s_1$ and $s_2$, we can identify a prefix and a suffix of $P$ that contains all swirly paths and are assigned to $s_1$ and $s_2$.
Between the prefix and suffix there is an interval that does not contain any swirly path from $s_1$ or $s_2$, and we use \cref{sec:ns-partition} to partition it.
By combining the prefix-suffix with the partition for the middle part, we obtain the desired partition for $P$ with respect to $\{s_1,s_2\}$.

Let $P_x=P^{left}_x$ (resp. $P_y$) be the prefix (resp. suffix) of $P$ computed by \cref{lem:out_partition_xy}.
We note that by definition of $(\{s_x,s_y\},\leftside)$-like, for every $v\in\Left(P,s_x)\cap\Left(P,s_y)$ we have $v\in P_x$ if and only if $\dist(s_x,v)<\dist(s_y,v)$.
The following lemma shows that for $s_x$, all the shortest paths to vertices of $\Left(P_x,s_x)$ are non-swirly.

\begin{lemma}\label{lem:sx-to-Px-non-swirly}
    If $v\in \Left(P_x,s_x)$ then $R_{s_x,v}$ is a non-swirly path.
    Similarly, if $v\in \Left(P_y,s_y)$ then $R_{s_y,v}$ is a non-swirly path.
\end{lemma}
\begin{proof}
    We prove the first claim, the proof of the second claim is symmetric.
    If $v$ appears to the right of $x$ in the shortest paths tree of $s_x$, then the lemma follows from \cref{lem:rightloseleftswirl}.
    If $v$ appears to the left of $x$ in the shortest paths tree of $s_x$, then by \cref{lem:rightloseleftswirl} we have that $v\in\Left(P.s_y)$ and $\dist(s_x,v)>\dist(s_y,v)$, contradicting the assumption that $v\in P_x$.
\end{proof}

Recall that a path $R_{s,v}$ from $s\in F_{\leftside}$ to $v\in \Left(P,s)$ is swirly if it crosses either $R_x$ or $R_y$.
The following lemma states that a swirly path must cross $R_x$ or $R_y$ from left to right, and must intersect the second path (among $R_x$ and $R_y$).

\begin{lemma}\label{lem:swirl-cross-and-intersect}
    Let $s\in F_{\leftside}\setminus (P\cup\{s_x,s_y\})$, and let $v\in\Left(P,s)$.
    If $R_{s,v}$ is an $x$-swirly path, then $R_{s,v}$ crosses $R_x$ from left to right (when considering $R_x$ as being oriented from $x$ to $s_x$).
    Moreover, $R_{s,x}$ must intersects $R_y$.

    Similarly, if $R_{s,v}$ is a $y$-swirly path, then $R_{s,v}$ crosses $R_y$ from left to right (when considering $R_y$ as being oriented from $s_y$ to $y$).
    Moreover, $R_{s,y}$ must intersects $R_x$.
\end{lemma}
\begin{figure}[ht]
    \centering
    \includegraphics[width=0.3\textwidth]{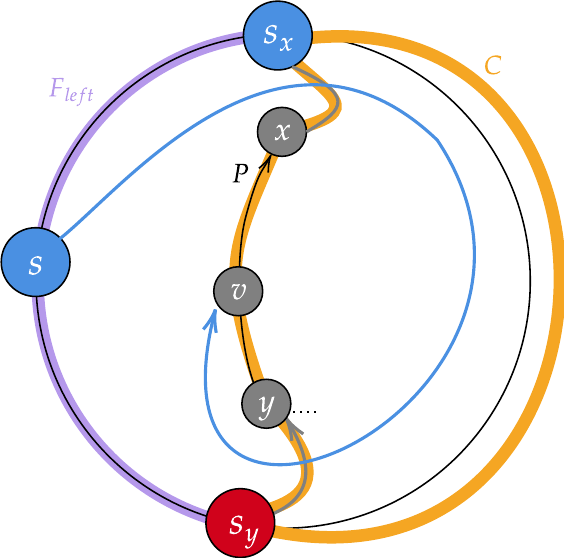}
    \caption{The cycle $C$ in the proof of \cref{lem:swirl-cross-and-intersect}.}
    \label{fig:swirl-cross-and-intersect}
\end{figure}
\begin{proof}
    We prove the first statement, the proof of the second statement is symmetric.
    Consider the cycle $C$ obtained by concatenating $R_y$, $P$,$R_x$ and the Jordan curve connecting $s_x$ and $s_y$ embedded in the face $F$.
    Notice that $R_x\cap R_y=\emptyset$ (since each one of them belongs to a different Voronoi cell) and therefore $C$ is a non-self-crossing cycle.
    We think of $C$ as being oriented consistently with $P$ so as to define left and right properly.
    Notice that all edges entering $P\setminus\{x,y\}$ from the left are on the left side of $C$.
    Since $s\in F_{\leftside}\setminus (P\cup\{s_x,s_y\})$ it must be on the left side of $C$.
    Thus, when a path starting at $s$ first crosses $C$ it must cross $C$ from left to right.
    By uniqueness of shortest paths $R_{s,v}$ does not cross $P$, and since $R_{s,v}$ is an $x$-swirly path, the first time $R_{s,v}$ crosses $C$ it must cross $R_x$.
    Thus, $R_{s,v}$ crosses $R_x$ from left to right.

    Let $w$ be the first vertex of $R_{s,v}$ strictly in the right side of $C$ ($w$ exists since $R_{s,v}$ crosses $C$).
    Since $R_{s,v}$ first crosses $R_x$, we have that $R_{s,v}[s,w)\cap R_x\ne \emptyset$.
    Therefore, by uniqueness of shortest paths, $R_{s,v}[w,v]$ cannot intersect $R_x$.
    All the edges entering $P$ on the right side of $C$ enter $P$ from the right.
    Therefore, $R_{s,v}[w,v]\cap R_y$ must be non-empty, since otherwise $v\notin \Left(P,s)$.
\end{proof}

The following lemma and corollary show that a swirly path $R_{s,v}$ is $y$-swirly if and only if $v\in P_x$.

\begin{lemma}\label{lem:s-x-swirl}
    Let $s\in F_{\leftside}\setminus\{s_x,s_y\}$ and let $v\in \Left(P,s)$ such that $R_{s,v}$ is an $x$-swirly path.
    Then, $v$ appears to the left of $x$ in the shortest paths tree of $s_x$, $v\in\Left(P,s_x)$, $v\in \Left(P,s_y)$ and $\dist(s,v)>\dist(s_x,v)\ge\dist(s_y,v)$.

    Similarly, if $R_{s,v}$ is a $y$-swirly path, then, $v$ appears to the right of $y$ in the shortest paths tree of $s_y$, $v\in\Left(P,s_y)$, $v\in \Left(P,s_x)$ and $\dist(s,v)>\dist(s_y,v)\ge\dist(s_x,v)$. .
\end{lemma}
\begin{proof}
We prove the first statement, the proof of the second statement is symmetric.
By \cref{lem:swirl-cross-and-intersect}, $R_{s,v}$ crosses $R_x$ from right to left (when considering $R_x$ oriented from $s_x$ to $x$).
For simplicity of the proof, we assume that $R_{s,v}\cap R_x$ contains exactly one vertex, $z$.
Let $w$ be the last (farthest from $s_x$) vertex on $R_{s_x,v}$ that is on $R_x$.
We first note that if $w=z$ then $R_{s_x,v}[w,v]=R_{s,v}[z,v]$ and therefore the lemma holds.
Moreover, by uniqueness of shortest paths, it cannot be the case that $w$ appears after $z$ on $R_x$.
Thus, it remains to consider the case where $w$ appears strictly before $z$ on $R_x$.
For this case, we prove the two claims separately.

\begin{figure}[ht]
    \centering
    \begin{subfigure}[t]{0.3\textwidth}
        \centering
        \includegraphics[width=\textwidth]{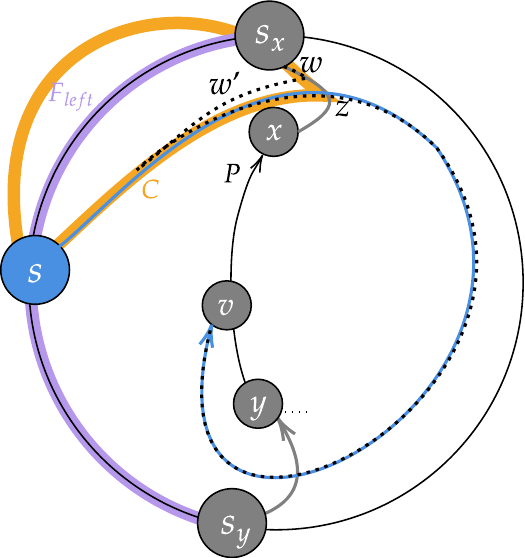}
        \caption{The cycle $C$. The dotted black path is $R_{s_x,v}$.}
        \label{fig:Swirl-S-SX-CA}
    \end{subfigure}
    \hspace{1in}
    \begin{subfigure}[t]{0.3\textwidth}
        \centering
        \includegraphics[width=\textwidth]{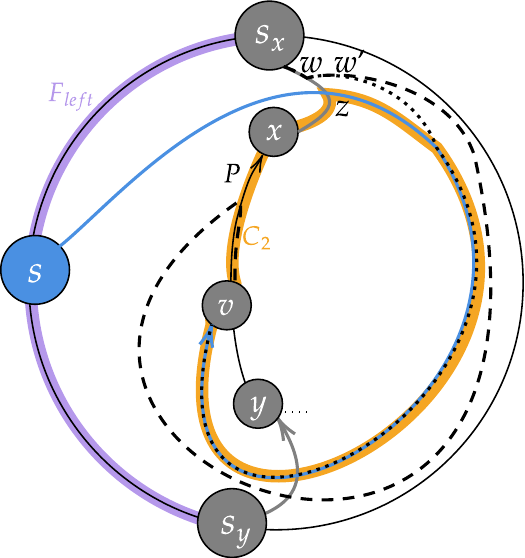}
        \caption{The cycle $C_2$. The dotted and the dashed black paths are the two options for $R_{s_x,v}$.}
        \label{fig:Swirl-S-SX-C2A}
    \end{subfigure}
    \caption{Figures for the proof of  \cref{lem:s-x-swirl}.  }
    \label{fig:s-x-swirl}
\end{figure}

If $v$ appears to the left of $x$ in the shortest paths tree of $s_x$.
Let $w'$ be the following vertex of $w$ on $R_{s_x,v}$.
Assume by contradiction that $v$ appears to the right of $x$.
This means that $R_{s_x,v}$ exits $R_x$ to the right.
Consider the cycle $C$ composed of the concatenation of $R_x[s_x,z]$, $R_{v,s}[z,s]$ and the Jordan curve connecting $s$ and $s_x$ embedded in the face $F$ (see \cref{fig:Swirl-S-SX-CA}).
We think of $C$ as oriented consistently with $R_x$ (from $s_x$ to $z$).
Clearly, $w'$ is in the right side of $C$.
On the other hand, $v$ is in the left side $C$.
To see this, consider the path $R_{s,v}[z,v]$.
Notice that by uniqueness of shortest paths $R_{s_x,v}[w',v]$ cannot cross $R_x$.
Assume $R_{s_x,v}[w',v]$ intersects with $R_{s,v}[s,z]$, then by uniqueness of shortest paths, it must reach $z$, thus crosses $R_x$, a contradiction.

If $v\in\Left(P,s_x)$.
Recall that we consider the case where $w$ is strictly before $z$ on $R_x$ and $w'$ is to the left of $R_x$.
Consider the cycle $C_2$ composed of the concatenation of $P[v,x]$, $R_x[x,z]$, $R_{s,v}[z,v]$ (see \cref{fig:Swirl-S-SX-C2A}).
We think of $C$ as an oriented cycle whose orientation is consistent with $P$.
Notice that all the right-entering arcs of $P\setminus (R_x\cap R_{s,v})$  are in the right side of $C_2$ and that $w$ is on the left side of $C_2$.
Since $v$ is on the cycle and $w$ is strictly to the left of the cycle, $R_{s_x,v}[w',v]$ must intersect $C_2$ at some point, let $u$ be the first vertex of $R_{s_x,v}[w',v]$ on $C_2$.
Clearly, $u\notin R_x[z,x]$, by uniqueness of shortest paths.
If $u\in R_{s,v}$ then by uniqueness of shortest paths $s_x$ reaches $v$ from the left.
Otherwise, if $u\in P\setminus (R_x\cup R_{s,v})$, then $R_{s_x,v}$ enters $P$ in $u$ from the left.
Since $v\in\Left(P,s_x)$ and $v$ appears to the left of $x$ in the shortest paths tree of $s_x$, by \cref{lem:rightloseleftswirl} we have that $v\in\Left(P,s_y)$, as required.
Moreover, by \cref{lem:rightloseleftswirl} we also have $\dist(s_x,v)>\dist(s_y,v)$.
Finally, by the triangle inequality we have $\dist(s,v)=\dist(s,z)+\dist(z,v)>\dist(s_x,z)+\dist(z,v)\ge \dist(s_x,v)$, as required.
\qedhere

\end{proof}

The following is a direct consequence of \cref{lem:rightloseleftswirl,lem:s-x-swirl}.

\begin{corollary}\label{cor:Px-not-x-swirly}
Let $s\in F_{\leftside}\setminus\{s_x,s_y\}$ and let $v\in \Left(P_x,s)$.
If $R_{s,v}$ is a swirly path, then it is a $y$-swirly path.
Furthermore, if $v\in \Left(P_x,s_x)$ then $R_{s_x,v}$ is a non-swirly path.

Similarly for $v\in \Left(P_y,s)$, if $R_{s,v}$ is a swirly path, then it is an $x$-swirly path.
Furthermore, if $v\in \Left(P_y,s_y)$ then $R_{s_y,v}$ is a non-swirly path.
\end{corollary}

In the following Lemmas (\cref{clm:being-left-to-vtop-is-monotone,lem:top-bottom-swirl-together,lem:classifyinterval,lem:TopTreeBotTree-Partition,lem:Partition-P-for-one-site}) we show that for any $s\in F_{\leftside}\setminus\{s_x,s_y\}$ we can either deduce that $s$ is not the site of any $v\in\Left(P,s)$ or a partition of $\Left(P,s)$ into an interval of non-swirly paths, and an interval of swirly paths.

\begin{lemma}\label{clm:being-left-to-vtop-is-monotone}
    Let $s\in F_{\leftside}\setminus P$ and $v_{top},v_{bottom}\in \Left(P,s)$ such that $v_{top}$ is closer to $x$ than $v_{bottom}$ and in the shortest paths tree of $s$, $v_{bottom}$ appears not to the left of $v_{top}$.
    Then, for every $u\in \Left(P[v_{top},v_{bottom}],s)$ we have that $u$ appears not to the left of $v_{top}$ in the shortest paths tree of $s$.
\end{lemma}
\begin{figure}[ht]
    \centering
    \includegraphics[width=0.3\textwidth]{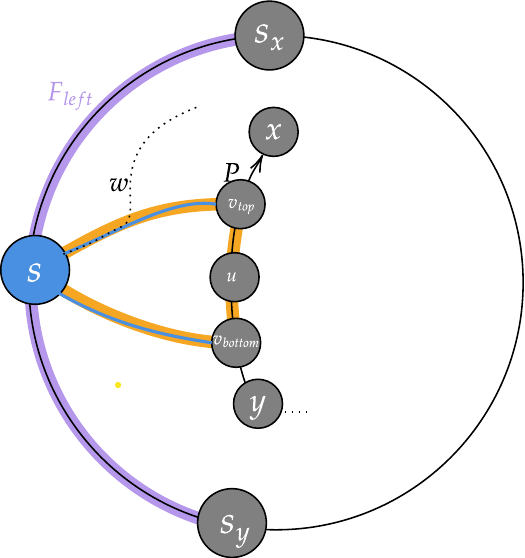}
    \caption{The cycle $C$ for the proof of \cref{clm:being-left-to-vtop-is-monotone}.}
    \label{fig:being-left-to-vtop-is-monotone}
\end{figure}
\begin{proof}
    Assume to the contrary that there is a vertex $u\in P[v_{top},v_{bottom}]$ that $s$ reaches $u$ from the left and $v$ is to the left of $v_{top}$ in the shortest paths tree of $s$.

    Consider the cycle $C$ composed of the concatenation of $R_{s,v_{bottom}}$, $P[v_{bottom},v_{top}]$, and $R_{v_{top},s}$ (see \cref{fig:being-left-to-vtop-is-monotone}).
    We think of $C$ as oriented consistently with $P$.
    Notice that $C$ has no self crossing, since it consists of three shortest paths with shared endpoints.
    Let $(w',w)$ be the first edge of $R_{s,u}$ that is not on $R_{s,v_{top}}$.
    Since $u$ is to the left of $v_{top}$ in the shortest paths tree of $s$, $w$ is on the right side of $C$ (notice that $R_{s,v_{top}}$ is reversely oriented in $C$).
    The path $R_{s,u}[w,u]$ ends on $C$, so it must intersect $C$.
    Let $z$ be the first vertex of $R_{s,u}[w,u]$ that is on $C$.
    By uniqueness of shortest paths, $z$ cannot be on $R_{s,v_{top}}$.
    Moreover, we claim that $z$ is not on $R_{s,v_{bottom}}$ as well.
    If $w$ is not on $R_{s,v_{bottom}}$, this follows from uniqueness of shortest paths.
    Otherwise, we have that $R_{s,v_{bottom}}[s,w] = R_{s,u}[s,w]$ which implies that $u$ is not to the left of $v_{top}$ in the shortest paths tree of $s$, a contradiction.
    We conclude that $z$ is on $P\setminus (R_{s,v_{top}}\cup R_{s,v_{bottom}})$.
    Recall that all left edges of $P\setminus (R_{s,v_{top}}\cup R_{s,v_{bottom}})$ are on the left side of $C$.
    It follows that $s$ reaches $u$ from the right, a contradiction.
\end{proof}

A subpath of $P[a,b]$ is called $(s,\leftside)$ swirly (resp. $x$-swirly, $y$-swirly, non-swirly) if for every $v\in \Left(P[a,b],s)$, the path $R_{s,v}$ is a swirly (resp. $x$-swirly, $y$-swirly, non-swirly) path.

\begin{lemma}\label{lem:top-bottom-swirl-together}
    Let $s\in F_{\leftside}\setminus (P\cup \{s_x,s_y\})$ and $v_{top},v_{bottom}\in \Left(P,s)$ such that $v_{top}$ is closer to $x$ than $v_{bottom}$ and  $v_{bottom}$ appears not to the left of $v_{top}$ in the shortest paths tree of $s$, then   $P[v_{top},v_{bottom}]$  is either a $(s,\leftside)$ non-swirly, a $(s,\leftside)$ $x$-swirly or a $(s,\leftside)$ $y$-swirly subpath.

\end{lemma}
\begin{proof}
Let $u\in \Left(P[v_{top},v_{bottom}],s)$ be a vertex.
We prove that  $R_{s,v_{top}}$ and  $R_{s,u}$ are both of the same type (non-swirly, $x$-swirly, or $y$-swirly).
By \cref{clm:being-left-to-vtop-is-monotone} we have that $u$ appears not to the left of $v_{top}$ in the shortest paths tree of $s$.
We distinguish between three cases:\\

\begin{figure}[ht]
    \centering
    \begin{subfigure}[t]{0.3\textwidth}
        \centering
        \includegraphics[width=\textwidth]{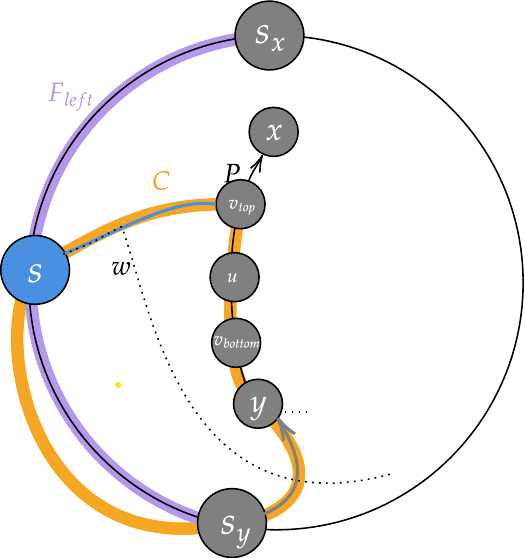}
        \caption{The cycle $C$.}
        \label{fig:Swirl-S-SX-C1}
    \end{subfigure}
     \hspace{0.2in}
    \begin{subfigure}[t]{0.3\textwidth}
        \centering
        \includegraphics[width=\textwidth]{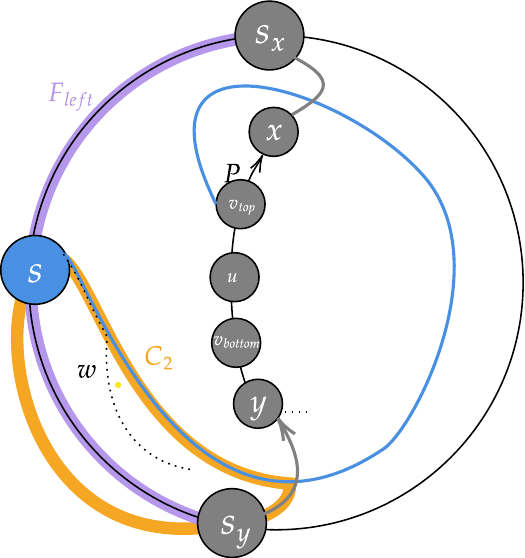}
        \caption{The cycle $C_2$.}
        \label{fig:Swirl-S-SX-C2}
    \end{subfigure}
     \hspace{0.2in}
    \begin{subfigure}[t]{0.3\textwidth}
        \centering
        \includegraphics[width=\textwidth]{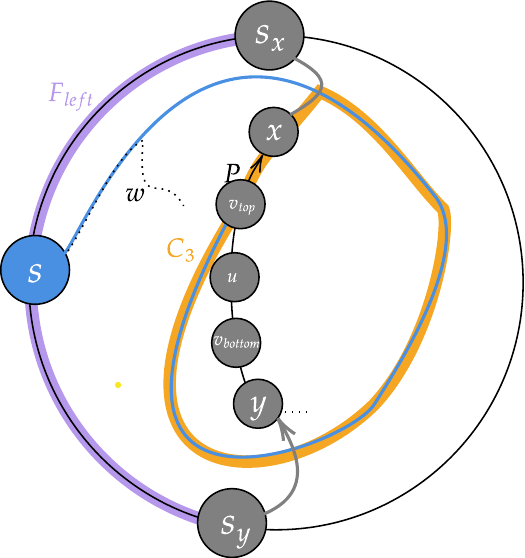}
        \caption{The cycle $C_3$.}
        \label{fig:Swirl-S-SX-C3}
    \end{subfigure}
    \caption{Figurs for the proof of  \cref{lem:top-bottom-swirl-together}.}
    \label{fig:top-bottom-swirl-togetherl}
\end{figure}

1. If $R_{s,v_{top}}$ is non-swirly.
Consider the cycle $C$ obtained by concatenating $R_{s,v_{top}},  P[v_{top},y], R_y$  and the Jordan curve connecting $s_y$ and $s$ embedded in the face $F$. See \cref{fig:Swirl-S-SX-C1}.
We think of $C$ as oriented consistently with that of $P$.
Notice that $R_{s,v_{top}}$ and $R_y$ do not cross each other since $R_{s,v_{top}}$ is a non-swirly path, therefore $C$ is non-self crossing.
Moreover, $R_{s,v_{top}}$ and $R_x$ do not cross each other either, and therefore $R_x$ is fully contained in the (weakly-)right side of $C$.
Let $(w',w)$ be the first edge of $R_{s,u}$ that is not on $R_{s,v_{top}}$.
If $w'\in P$ then $R_{s,u}=R_{s,v_{top}}[s,w']\circ P[w',u]$ is non-swirly.
Otherwise, if $w'\notin P$, since $u$ is not to the left of $v_{top}$ in the shortest paths tree of $s$, $w$ is on the left side of $C$ (notice that $R_{s,v_{top}}$ is reversely oriented in $C$).
Assume by contradiction that $R_{s,u}$ is a swirly path.
Thus, $R_{s,u}$ must cross either $R_x$ or $R_y$.
Since $R_x$ is in the right side of $C$, $R_{s,u}$ must cross $C$ to cross $R_x$ or $R_y$ (for $R_y$ it is straightforward).
By uniqueness of shortest paths $R_{s,u}[w',u]$ cannot cross $P$ or $R_{s,v_{top}}$ thus it must cross $R_y$. Let $z$ be the first vertex of $R_{s,u}$ after $R_y$.
However, since the left edges of $P[y,v_{top}]\setminus (R_{s,v_{top}}\cup R_y)$ are all in the left side of $C$, and $u\in \Left(P[y,v_{top}],s)$ it must be that $R_{s,u}[z,u]$ crosses $C$.
By uniqueness of shortest paths $R_{s,u}[z,u]$ does not cross $R_{s,v_{top}}$, $P$, or $R_y$, n a contradiction.\\

2. If $R_{s,v_{top}}$ is $y$-swirly.
For simplicity of the proof, we assume that $R_{s,v_{top}}\cap R_y$ contains a single vertex, $z$.
In the general case, it may be a continuous subpath but the proof is similar.
Consider the cycle $C_2$ composed of the concatenation of $R_{s,v_{top}}[s,z]$, $R_y[s_y,z]$,  and the Jordan curve connecting $s_y$ and $s$ embedded in the face $F$. See \cref{fig:Swirl-S-SX-C2}.
We think of $C_2$ as oriented  consistently with that of $R_y$ from $s_y$ to $z$.
Notice that $C_2$ is non self-crossing, and that both $P$ and $R_x$ are on the right side of $C_2$.

Let $(w',w)$ be the first edge of $R_{s,u}$ that is not on $R_{s,v_{top}}$.
If $w'$ appears after $z$ in $R_{s,v_{top}}$ then $R_{s,u}$ crosses $R_y$ at $z$ (and before crossing $R_x$), which means that $R_{s,u}$ is $y$-swirly.
Otherwise, if $w'$ appears before $z$ on $R_{s,v_{top}}$, since $u$ is not to the left of $v_{top}$ in the shortest paths tree of $s$, $w$ is on the left side of $C_2$ (notice that $R_{s,z}$ is reversely oriented in $C_2$).
Thus, in order to reach $u\in P$ $R_{s,u}[w,u]$ must cross $C_2$.
By uniqueness of shortest paths,  $R_{s,u}[w,u]$ cannot cross $R_{s,v_{top}}$, therefore it crosses $R_y$.
Since $R_{s,u}$ cannot cross $R_x$ before crossing $C_2$, we conclude that $R_{s,u}$ is indeed a $y$-swirly path.\\

3.  If $R_{s,v_{top}}$ is an $x$-swirly path.
For simplicity of the proof, we assume that $R_{s,v_{top}}\cap R_x$ contains a single vertex, $z$.
In the general case, it may be a continuous subpath but a similar proof still exists.
Consider the cycle $C_3=R_{s,v_{top}}[z,v_{top}]\circ P[v_{top},x]\circ R_x[x,z]$. See \cref{fig:Swirl-S-SX-C3}.
We think of $C_3$ as an oriented cycle whose orientation is consistent with that of $P$ so as to define left and right properly.
Notice that $C_3$ has no self crossing, and that the left side of $P[y,v_{top}]\setminus R_{s,v_{top}}$ are on the right side of $C_3$.

Let $(w',w)$ be the first edge of $R_{s,u}$ that is not on $R_{s,v_{top}}$.
If $w'$ appears after $z$ in $R_{s,v_{top}}$ then $R_{s,u}$ crosses $R_x$ at $z$ (and before crossing $R_y$), which means $R_{s,u}$ is $x$-swirly.
Otherwise, $w'$ is on the left side of $C_3$.
Thus, in order to reach $u\in P[y,v_{top})$, $R_{s,u}[w,u]$ must cross $C_3$.
By uniqueness of shortest paths  $R_{s,u}[w,u]$ cannot cross either $R_{s,v_{top}}$ or $P$, therefore, it must cross $R_x$.
Moreover, it must cross $R_x$ from left to right, and therefore $R_{s,u}$ must cross $R_x$ before crossing $R_y$.
Thus, $R_{s,u}$ is indeed $x$-swirly path. \qedhere
\end{proof}

\begin{lemma}\label{lem:classifyinterval}
    There exists an algorithm that given $s\in F_{\leftside}\setminus (P\cup\{s_x,s_y\})$ and $v_{top},v_{bottom}\in P$ such that $v_{top}$ is closer to $x$ than $v_{bottom}$ and in the shortest paths tree of $s$, $v_{top}$ appears to the left of $v_{bottom}$, reports in $\Otild(1)$ time one of the following:
    \begin{enumerate}
        \item Reports that $P[v_{top},v_{bottom}]$ is an $(s,\leftside)$ non-swirly subpath.
        \item Reports that for every $v\in \Left(P[v_{top},v_{bottom}],s)$ we have, $\dist(s,v)>\min\{\dist(s_x,v),\dist(s_y,v)\}$.
    \end{enumerate}
\end{lemma}
\begin{proof}
    The algorithm first computes a partition of $P$ into $P_x$ and $P_y$, using \cref{lem:out_partition_xy}.
    Then, the algorithm checks if $P[v_{top},v_{bottom}]$ is fully contained in either $P_x$ or $P_y$.
    If it is not, the algorithm reports Case 1.
    Otherwise, if $P[v_{top},v_{bottom}]\subseteq P_x$ (the case where $P[v_{top},v_{bottom}]\subseteq P_y$ is treated symmetrically), the algorithm checks if $v_{bottom}\in\Left(P,s_x)$ .
    If $v_{bottom}\notin\Left(P,s_x)$  or $\dist(s,v_{bottom})<\dist(s_x,v_{bottom})$, the algorithm reports Case 1.
    Otherwise, if $v_{bottom}\in\Left(P,s_x)$ and $\dist(s,v_{bottom})>\dist(s_x,v_{bottom})$, the algorithm reports Case 2.
    Clearly, the running time of the algorithm is $\Otild(1)$.

    \medskip
\noindent
{\bf Correctness.}
    If the algorithm reports Case 1 due to $P[v_{top},v_{bottom}]$ being not fully contained in either $P_x$ or $P_y$.
    In this case, it must be that $v_{top}\in P_x$ and $v_{bottom}\in P_y$.
    However, by \cref{cor:Px-not-x-swirly} it must be that $R_{s,v_{top}}$ is not an $x$-swirly path, and $R_{s,v_{bottom}}$ is not a $y$-swirly path.
    The only remaining option according to \cref{lem:top-bottom-swirl-together}, is that both $R_{s,v_{top}}$  and $R_{s,v_{bottom}}$ are non-swirly paths.

    Otherwise, we have $P[v_{top},v_{bottom}]\subseteq P_x$.
    If the algorithm reports Case 1 due to $s_x$ reaches $v_{bottom}$ from the right or $\dist(s,v_{bottom})<\dist(s_x,v_{bottom})$, by \cref{cor:Px-not-x-swirly} it must be that $R=R_{s,v_{bottom}}$ is not a $y$-swirly path.
    Since $v\in P_x$, $R$ is also not an $x$-swirly path, thus $R$ is a non-swirly path.
    Thus, for every $v\in P[v_{top},v_{bottom}]$, by \cref{lem:top-bottom-swirl-together}, we have that $R_{s,v}$ is also a non-swirly path.

    Finally, if the algorithm reports Case 2 due to $v_{bottom}\in\Left(P,s_x)$ and $\dist(s_x,v_{bottom})<\dist(s,v_{bottom})$, there are two cases, based on the type of $R_{s,v_{bottom}}$ (which the algorithm never recognizes).
    \begin{enumerate}
        \item If $R_{s,v_{bottom}}$ is a $y$-swirly path, then by \cref{lem:top-bottom-swirl-together}, for every $v\in \Left(P[v_{top},v_{bottom}],s)$ we have that $R_{s,v}$  is $y$-swirly.
        Thus, by \cref{lem:s-x-swirl} we have  $\dist(s,v)>\dist(s_x,v)$ .

        \item If $R_{s,v_{bottom}}$ is a non-swirly path, then by \cref{lem:top-bottom-swirl-together} for every $v\in \Left(P[v_{top},v_{bottom}],s)$ we have that $R_{s,v}$  is a non-swirly path.
        Moreover, by \cref{lem:sx-to-Px-non-swirly}, for every $v\in \Left(P[v_{top},v_{bottom}],s_x)$ we have $R_{s_x,v}$  is a non-swirly path.
        Thus, by \cref{clm:winner_takes_something}, since $\dist(s_x,v_{bottom})<\dist(s,v_{bottom})$ and $v_{bottom}\in\Left(P,s_x)$, we have that every $u\in P[v_{top},v_{bottom}]$ is a vertex that $(\{s_x,s\},\leftside)$-likes $s_x$.
        In particular, for every $v\in \Left(P[v_{top},v_{bottom}]s)$, we have $\dist(s,v)>\min\{\dist(s,v),\dist(s_x,v),\dist(s_y,v)\}$ which means $\dist(s,v)>\min\{\dist(s_x,v),\dist(s_y,v)\}$ as required.\qedhere
    \end{enumerate}
\end{proof}

\begin{lemma}\label{lem:TopTreeBotTree-Partition}
There is an algorithm that given $s\in F_{\leftside} \setminus (P\cup\{s_x,s_y\})$ outputs in $\Otild(1)$ time at most two subpaths $P[a_1,b_1]$ and $P[a_2,b_2]$ such that all vertices in $\Left(P,s)\subseteq P[a_1,b_1] \cup P[a_2,b_2]$, and $b_1$ (resp. $b_2$) is not to the left of $a_1$ (resp. $a_2$) in the shortest paths tree of $s$.
In addition, $s$ reaches $a_1,b_1,a_2,b_2\in \Left(P,s)$ and $a_1$ is closer to $x$ than $a_2$.
Moreover, if one of the two subpaths is non-swirly, then the other one is swirly.
\end{lemma}
\begin{proof}
     Let $v_{top}$ and $v_{bottom}$ respectively be the first (closest to $x$) vertex on $P$ and the last (closest to $y$) vertex on $\Left(P,s)$.
     The algorithm starts by finding $v_{top}$ and $v_{bottom}$ using $\countA$ and $\select$ queries.

    The algorithm then checks if $v_{bottom}$ does not appear to the left of $v_{top}$ in the shortest paths tree of $s$, then $P[v_{top},v_{bottom}]$  is reported as one interval.

    Otherwise, the algorithm executes a binary search on the range $\Left(P[v_{top},v_{bottom}],s)$ for the last (furthest from $x$) vertex $v'$ such that $v'$ does not appear to the left of $v_{top}$ in the shortest paths tree of $s$.
    Notice that by \cref{clm:being-left-to-vtop-is-monotone} a binary search indeed finds this last vertex.

    Let $v''$ be first vertex after $v'$ on $\Left(P,s)$.
    The algorithm reports $P[a_1,b_1]=P[v_{top},v']$ and $P[a_2,b_2]=P[v'',v_{bottom}]$.
    Clearly, every vertex that $s$ reaches from the left is in $P[a_1,b_2]\cup P[a_2,b_2]$.
    Moreover, by definition of $v'$ we have that $b_1$ is not to the left of $a_1$ in the shortest paths tree of $s$.
    We prove that $b_2=v_{bottom}$ is not to the left of $a_2=v''$ in the shortest paths tree of $s$.
    Assume by contradiction that $v_{bottom}$ is to the left of $v''$ in the shortest paths tree of $s$.
    Consider the cycle $C$ composed of the concatenation of $R_{s,v''}$, $P[v'',v_{top}]$ and $R_{v_{top},s}$.
    We think of $C$ as oriented consistently with $P$.
    Notice that $C$ has no self-crossing.
    Let $(w',w)$ be the first edge on $R_{s,v_{bottom}}$ that is not $R_{s,v''}$.
    Due to our assumption that $v_{bottom}$ appears to the left of $v''$, such an edge must exist.
    Since $v''$ is to the left of $v_{top}$, it also must be that $w\notin R_{s,v_{top}}$.
    By uniqueness of shortest paths, $R_{s,v_{bottom}}[w,v_{bottom}]$ cannot intersect neither $R_{s,v_{top}}$ nor $R_{s,v''}$.
    On the other hand, one can prove that all the vertices of $P$ are not on the left side of $P$.
    In particular, $v_{bottom}\in P$ is not to the left of $C$.
    Since $R_{s,v_{bottom}}[w,v_{bottom}]$ cannot cross $P$, it must be that $v_{bottom}\in P(v'',v_{top})$.
    Thus, $v''$ is in $\Left(P,s)$ and is closer to $y$ than $V_{bottom}$, contradicting the definition of $v_{bottom}$.
    We conclude that $b_2=v_{bottom}$ is not to the left of $a_2=v''$.

\medskip
\noindent
{\bf If one of the two subpaths is non-swirly, then the other is a swirly subpath.}
To prove the suffix of the lemma, by \cref{lem:top-bottom-swirl-together} it is enough to show that if one of $R_{s,v_{top}}$ and $R_{s,v_{bottom}}$ is non-swirly, then the other is swirly.
Assume w.l.o.g. that $R_{s,v_{top}}$ is a non-swirly path.
Consider the cycle $C_2$ obtained by concatenating $R_{s,v_{top}}, P[v_{top},x], R_x,$  and the Jordan curve connecting $s_x$ and $s$ embedded in the face $F$.
We think of $C$ as an oriented consistently with $P$.
Notice that $R_{s,v_{top}}$ and $R_x$ do not cross each other since $R_{s,v_{top}}$ is a non-swirly path, therefore $C_2$ is non-self crossing.
Recall that $v_{bottom}$ appears to the left of $v_{top}$ in the shortest paths tree of $s$.
Consider the first edges $(w',w)$ on $R_{s,v_{bottom}}$ that is not on $R_{s,v_{top}}$, it must be that $w$ is to the left of $C_2$.
On the other hand, $v_{bottom}$ is on the right side of $C_2$, as one can see by following the path $R_{s,v_{top}}\circ P[v_{top},v_{bottom}]$.
Thus, $R_{s,v_{bottom}}[w,v_{bottom}]$ must cross $C_2$.
However, by uniqueness of shortest paths, $R_{s,v_{bottom}}$ does not cross nor $R_{s,v_{top}}$ nor $P$, therefore it must cross $R_x$, which means that $R_{s,v_{bottom}}$ is indeed a swirly path, as required.
\end{proof}

\begin{lemma}\label{lem:Partition-P-for-one-site}
    There exists an algorithm that given $s\in F_{\leftside}\setminus (P\cup\{s_x,s_y\})$ outputs in $\Otild(1)$ time one of the following:
    \begin{enumerate}
        \item At most two subpaths $P[a_1,b_1]$ and $P[a_2,b_2]$ such that $\Left(P,s)\subseteq P[a_1,b_1] \cup P[a_2,b_2]$, and $a_1,b_1,a_2,b_2\in\Left(P,s)$, and flags $f_1,f_2 \in \{‘N’,‘X’,‘Y’\} $ such that:
        \begin{enumerate}
        \item If $f_i=‘N’$, then $P[a_i,b_i]$ is an $(s,\leftside)$ non-swirly subpath.
        \item If $f_i=‘Y’$, then $P[a_i,b_i]$ is an $(s,\leftside)$ $y$-swirly subpath.
        \item If $f_i=‘X’$,then $P[a_i,b_i]$ is an $(s,\leftside)$ $x$-swirly subpath.
    \end{enumerate}
        \item Reports that for every $v\in \Left(P,s)$, we have $\dist(s,v)>\min\{\dist(s_x,v)\dist(s_y,v)\}$.
    \end{enumerate}
\end{lemma}
\begin{proof}
    The algorithm first computes a partition of $P$ into $P_x$ and $P_y$, using \cref{lem:out_partition_xy}.
    Then, the algorithm applies \cref{lem:TopTreeBotTree-Partition} to obatin at most two subpaths $P[a_1,b_1]$ and $P[a_2,b_2]$ such that all vertices that $s$ reaches from the left on $P$ are in $P[a_1,b_1] \cup P[a_2,b_2]$, and $b_1$ (resp. $b_2$) is not to the left of $a_1$ (resp. $a_2$) in the shortest paths tree of $s$.
    Moreover, if one of the two subpaths is non-swirly, then the other one is swirly.

    First, consider the case where \cref{lem:TopTreeBotTree-Partition} reports only one subpath, $P[a_1,b_1]$.
    In this case, notice that $P[a_1,b_1]$ is a valid input for \cref{lem:classifyinterval}.
    If \cref{lem:classifyinterval} finds that $P[a_1,b_1]$ is an $(s,\leftside)$ non-swirly subpath, the algorithm reports $f_1=‘N’$. Otherwise, the algorithm reports Case 2.
    Observe that every vertex outside of $P[a_1,b_1]$ is reached by $s$ from the right, so it trivially satisfies any of the possible classifications.

    In the other case, we obtain two disjoint subpaths $P[a_1,b_1]$ and $P[a_2,b_2]$ from \cref{lem:TopTreeBotTree-Partition}.
    Each one of those subpaths is a valid input for \cref{lem:classifyinterval}.
    So, the algorithm applies \cref{lem:classifyinterval} on both subpaths.
    \begin{enumerate}
        \item
    If both subpaths are found to be in Case 2 of \cref{lem:classifyinterval} (i.e. for every $v\in \Left(P[a_i,b_i],s)$  we have, $\dist(s,v)>\min\{\dist(s_x,v),\dist(s_y,v)\}$) we report Case 2 (i.e. that that for every $v\in \Left(P,s)$, we have $\dist(s,v)>\min\{\dist(s_x,v),\dist(s_y,v)\}$).
    This reports indeed holds, since $\Left(p,s)\subseteq P[a_1,b_1]\cup P[a_2,b_2]$.

    \item If one of the subpaths is found to be $(s,\leftside)$ non-swirly, it is guaranteed (by \cref{lem:TopTreeBotTree-Partition}) that the other one is swirly.
    Let $P_j=P[a_j,b_j]$ be the subpath that is $(s,\leftside)$ swirly. By \cref{cor:Px-not-x-swirly}, it must be that $P_j$ is contained in either $P_x$ or $P_y$.
    Thus, the algorithm checks whether $P_j\subseteq P_x$ or $P_j\subseteq P_y$ and deduces whether the corresponding subpath of $P$ is $(s,\leftside)$ $x$-swirly or $(s,\leftside)$ $y$-swirly.
    Finally, the algorithm returns $P[a_1,b_1]$ and $P[a_2,b_2]$ and the corresponding flags (one is $‘N’$ and the other is $‘X’$ or $‘Y’$ as we describe).\qedhere
    \end{enumerate}
\end{proof}
In the following two lemmas, we show that we can utilize swirly paths in each of $s_1$ and $s_2$ to deduce that a prefix or suffix of $P$ can be assigned in the relaxed partition with respect to $\{s_1,s_2\}$ with one site.
More specifically, the prefix (that contains $x$ can be associated with $s_1$ and the suffix can be associated with $s_2$.
\begin{figure}[ht]
    \centering
    \includegraphics[width=0.3\textwidth]{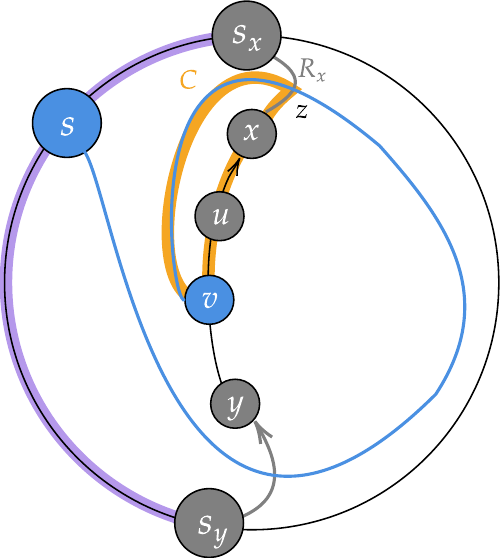}
    \caption{The cycle $C$ for the proof of \cref{lem:swirly-monotone}.}
    \label{fig:swirly-monotone}
\end{figure}
\begin{lemma}\label{lem:swirly-monotone}
    Let $s\in F_{\leftside}\setminus (P\cup \{s_x,s_y\})$ and let $v\in \Left(P,s)$ such that $R_{s,v}$ is a $y$-swirly path.
    Then for every $u\in \Left(P[x,v],s)$ we have that $R_{s,u}$ is also a $y$-swirly path.

    Similarly, if $R_{s,v}$ is an $x$-swirly path, then for every $u\in \Left(P[v,y],s)$ we have that $R_{s,u}$ is also an $x$-swirly path.
\end{lemma}
\begin{proof}
    We prove the first statement, the proof of the second statement is similar.
    Let $z$ be some vertex in $R_{s,v}\cap R_x$ ($z$ exists by \cref{lem:swirl-cross-and-intersect}).
    Consider the cycle $C$ obtained by concatenating $P[v,x], R_x[x,z],$  and $R_{s,v}[z,v]$ (see \cref{fig:swirly-monotone}).
    We think of $C$ as oriented consistently with $P$.
    Notice that $C$ is composed of three shortest paths and therefore $C$ is non-self crossing.
    Let $u\in\Left(P[x,v),s)$.
    Since $u$ is on $C$, $R_{s,u}$ must intersects $C$.
    Let $(w',w)$ be the first (closest to $s$) vertex of $R_{s,u}$ which is not on $R_{s,v}$.
    If $R_{s,v}[s,w']$ crosses $R_y$, then $R_{s,u}$ is an $x$-swirly path, as required.
    Otherwise, $w'$ is on the right side of $C$.
    Assume to the contrary that $R_{s,u}$ is a non-swirly path.
    We get a contradiction since $R_{s,u}$ is entirely on the right side of $C$ but all edges entering $P[x,v)$ from the left are in the left side of $C$.
    Assume to the contrary that $R_{s,u}$ is an $x$-swirly path, then by \cref{lem:swirl-cross-and-intersect} it must cross $R_x$ from left to right which means it first cross $R_x[s_x,z)$.
    Then, by uniqueness of shortest paths it cannot cross $C$ again, which means that $R_{s,u}$ reaches $u$ from the left, a contradiction.
\end{proof}

\begin{lemma}\label{lem:s1-swirl-s2-cannot-win}
Let $s_1$ and $s_2$ be two vertices on $F_{\leftside}$ such that $s_1$ is closer to $s_x$ than $s_2$ on $F_{\leftside}$.
For every $w\in \{1,2\}$ and $v\in P_x$, if $s_w$ reaches $v$ via a $y$-swirly path, every vertex on $P_x[a_x,v]$ $(\{s_1,s_2\},\leftside)$-likes $s_1$.

Similarly, for every $w\in \{1,2\}$ and $v\in P_y$, if $s_w$ reaches $v$ via a $x$-swirly path, every vertex on $P_x[v,b_x]$ $(\{s_1,s_2\},\leftside)$-likes $s_2$.

\end{lemma}
\begin{figure}[ht!]
    \centering
    \includegraphics[width=0.3\textwidth]{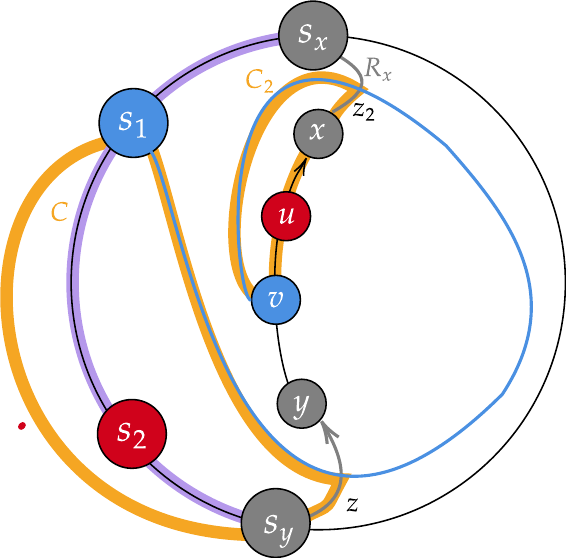}
    \caption{The cycles $C$ and $C_2$ for the proof of \cref{lem:s1-swirl-s2-cannot-win}.}
    \label{fig:s1-swirl-s2-cannot-win}
\end{figure}
\begin{proof}
    We prove the first statement, the proof of the second statement is symmetric.
    First, oversevere that the lemma follows for $w=2$, as the path $R_{s_2,v}$ crosses $R_x$, and therefore $\dist(s_x,v) < \dist(s_2,v)$, and $s_x \neq s_2$.
    It follows that $v$ indeed $(\{s_1,s_2\},\leftside)$-likes $s_1$.
    Moreover, by \cref{lem:swirly-monotone} the same argument holds for any $u\in\Left(P_x[a_x,v],s_2)$, and the claim is trivial for $u\in P_x[a_x,v]\setminus\Left(P_x[a_x,v],s_2)$.
    Thus, we focus on $w=1$ - i.e. $R_{s_1,v}$ is a $y$-swirly path.

For simplicity of the proof, we assume that $R_{s_1,v}\cap R_y$ contains a single vertex, $z$.
In the general case, it may be a continuous subpath but a similar proof still exists.
Consider the cycle $C$ composed of the concatenation of $R_{s_1,v}[s_1,z]$, $R_y[z,s_y]$  and the Jordan curve connecting $s_1$ and $s_y$ embedded in the face $F$.
Notice that $C$ has no self crossing, since $R_y[z,s_y]$ and $R_{s_1,v}[s_1,z]$ are both shortest paths outgoing from $z$.
Assume to the contrary that for some vertex $u\in P_x[a_x,v]$ that does not $(\{s_1,s_2\},\leftside)$-like $s_1$.
Consider the path $R_{s_2,u}$.
It follows from our assumption that $R_{s_2,u}$ enters $P$ from the left, and that $\dist(s_2,u) < \min_{s\in \{s_1,s_x,s_y\}}(\dist(s,u))$.
Notice that $s_2$ is on the left side of $C$ because, notice that $P$ and $u$ are on the right side of $C$.
Therefore, there must be a first vertex $c$ on $R_{s_2,u}$ that is on $C$.
Notice that $c$ cannot be on $R_x$, as that would lead to $\dist(s_x,u) < \dist(s_2,u)$, a contradiction.
We therefore have that $c\in R_{s_1,v}[s_1,z]$.

For simplicity of the proof, we assume that $R_{s_1,v}\cap R_x$ contains a single vertex, $z_2$.
In the general case, it may be a continuous subpath but a similar proof still exists.
Consider the cycle $C_2$ composed of the concatenation of $R_{s_1,v}[z_2,v]$, $P[v,x]$, $R_x[x,z_2]$ and the Jordan curve connecting $s_1$ and $s_y$ embedded in the face $F$.
We think of $C_2$ as oriented  consistently with $P$.
Notice that $C_2$ is non-self crossing.
Notice that $c$ is on the right side of $C_2$ and that every left edge of $P\setminus (R_{s_1,v} \cup R_y)$ is on the left side of $C_2$.
It follows that $R_{s_2,u}[c,u]$ must intersect $C_2$, and that the first intersection is not on $P \setminus (R_y \cup R_{s_1,v})$ (as this would lead to $R_{s_2,u}$ entering $P$ from the right).
Notice that $R_{s_2,u}[c,u]$ cannot intersect $R_y$, as this would lead to $\dist(s_y,u) < \dist(s_2,u)$, a contradiction.
We conclude that $R_{s_2,u}(c,u)$ intersects $C_2$ on $R_{s_1,v}[z_2,v]$ on some vertex $c_2$.
It follows from uniqueness of shortest paths that $R_{s_2,u}[c,c_2]= R_{s_1,v}[c,c_2]$.
This is a contradiction, as $R_{s_1,v}[c,c_2]$ crosses $R_y$ follows from the definition of $z$ and $z_2$.
\end{proof}

Finally, we are ready to prove \cref{lem:s1s2partition} by combining all the above.
This way we obtain a relaxed partition for any two sites $s_1.s_2\in F_{\leftside}$.
As described above, the idea is to assign a prefix and suffix of $P$ to $s_1$ and $s_2$ respectively, exploiting information on swirly paths.
Then the relaxed partition of $P$ is obtained by partition the remaining interval using the non-swirly algorithm (\cref{sec:ns-partition}).

\lemtwositespartition*
\begin{proof}
We distinguish between three cases.\\

1. If $s_1=s_x$ and $s_2=s_y$, the lemma is easily proved by \cref{lem:out_partition_xy}.\\

2. Next, consider the case where $\{s_1,s_2\}\cap \{s_x,s_y\}=\emptyset$.
In this case, the algorithm first computes a partition of $P$ into $P_x$ and $P_y$, using \cref{lem:out_partition_xy}.
Then, the algorithm applies \cref{lem:Partition-P-for-one-site} both on $s_1$ and on $s_2$.
If one of them reports Case 2 (i.e. that for every $v\in \Left(P,s)$ it holds $\dist(s,v)>\min\{\dist(s_x,v),\dist(s_y,v)\}$), the algorithm outputs a trivial partition of $P$ into $P_1=P$ and $P_2=\emptyset$ (if \cref{lem:Partition-P-for-one-site} reports Case 2 for $s_2$) or $P_1=\emptyset$ and $P_2=P$ (if \cref{lem:Partition-P-for-one-site} reports Case 2 for $s_1$).
Clearly, in this case it is a valid $(\{s_1,s_2\},\leftside)$-partition.

Consider the case where \cref{lem:Partition-P-for-one-site} reports Case 1 (at most two subpaths of $P$) for both $s_1$ and $s_2$.
Let $a\in\Left(P,s_1)\cup\Left(P,s_2)$ be the closest vertex to $y$ such that $R_{s_1,a}$ or $R_{s_2,a}$ is $y$-swirly.
Then, by \cref{lem:s1-swirl-s2-cannot-win} every vertex $v\in P[x,a]$ is $(\{s_1,s_2\},\leftside)$-likes $s_1$.
We note that if the vertex $a$ exists, it is the endpoint of one of the intervals reported by \cref{lem:Partition-P-for-one-site}.
If the vertex $a$ does not exist we consider $P[x,a]$ as an empty subpath.
Similarly, let $b\in \Left(P,s_1)\cup\Left(P,s_2)$ be the closest vertex to $y$ such that $R_{s_1,b}$ or $R_{s_2,b}$ is $x$-swirly.
Then, by \cref{lem:s1-swirl-s2-cannot-win} every vertex $v\in P[b,y]$ is $(\{s_1,s_2\},\leftside)$-likes $s_2$.
If the vertex $b$ does not exist we consider $P[b,y]$ as an empty subpath.
Finally, let $P[a',b']$ be the subpath of $P$ obtained by removing $P[x,a]$ (if $a$ exists) and $P[b,y]$ (if $b$ exists),
for every vertex $v\in\Left(P[a',b'],s_i)$ we have $R_{s_i,v}$ is non-swirly path for $i\in\{1,2\}$.
Thus, we apply \cref{lem:out_partition_non-swirly} on $\hat P=P[a',b']$ and obtain $\hat P_1$ and $\hat P_2$.
Finally, the algorithm returns $P_1=P[x,a]\circ \hat P_1$ and $P_2=\hat P_2\circ P[b,y]$.\\

3. Finally, consider the case $s_1=s_x$ and $s_2\ne s_y$ (the case $s_1\ne s_x$ and $s_2=s_y$ is symmetric).
This case is similar to the previous one, except that for now we have that $P_x$ is an $(s_1,\leftside)$-non-swirl subpath (by \cref{lem:sx-to-Px-non-swirly} and every $v\in P_y$ is $(\{s_1,s_2\},\leftside)$-likes $s_2$.
So the rest of the computation has to take care only on $P_x$.
In details, the algorithm applies \cref{lem:Partition-P-for-one-site} on $s_2$.
If it reports Case 2, the algorithm outputs a trivial partition $P_1=P$ and $P_2=\emptyset$.
Assume \cref{lem:Partition-P-for-one-site} reports Case 1, with the parts $P[a_1,b_1]$ and $P[a_2,b_2]$.
By \cref{cor:Px-not-x-swirly}, $P_x$ does not contain any $v$ with $R_{s_2,v}$ being $x$-swirly.
Let $P[a_i,b_i]$ be an $(s_2,\leftside)$ non-swirly subpath, and let $P_x=P[x,v]$.
If $a_i\notin P_x$ we return $P_1=P_x$ and $P_2=P_y$.
Otherwise, if $a_i\in P_x$ it must be that for every $v\in\Left(P[a_i,v],s_2)$ we have that $R_{s_2,v}$ is a non-swirly path.
To see this, $R_{s_2,v}$ cannot be a $x$-swirly path by \cref{cor:Px-not-x-swirly} since $v\in P_x$.
Assume to the contrary that $R_{s,v}$ is a $y$-swirly path, then we have $a_i\in P[x,v]$ and  by \cref{lem:swirly-monotone} we have $R_{s_2,a_i}$ is a $y$-swirly path, a contradiction.
Moreover, since $a_i$ is the closest to $x$ vertex in $\Left(P,s_2)$ with $R_{s_2,v}$ is non-swirly, we have that every $u\in P[x,a_i)$ is a vertex that $(\{s_1,s_2\},\leftside)$-likes $s_1$.
Therefore, we apply \cref{lem:out_partition_non-swirly} on $\hat P=P[a_i,v]$ and return $P_1=P[x,a_i)\circ \hat P_1$ and $P_2=\hat P_2\circ P(v,y]$.

In all three cases the running time of the algorithm is clearly $\Otild(1)$, and the correctness follows our discussion.
\end{proof}

\end{document}